\documentclass{llncs}

\usepackage{amssymb}
\usepackage{amsmath}
\usepackage{mathtools}
\usepackage{graphicx}
\usepackage{algorithm}
\usepackage{algorithmic}
\usepackage{cite}
\usepackage{setspace}
\usepackage{multirow}
\usepackage{array}
\usepackage{booktabs}
\usepackage{subeqnarray}
\usepackage{cases}

\title{Minimizing Deduction System and Its Application\thanks{This work was supported by National Natural Science Foundation (Grant 61572491, 61972297) and National Key Research and Development Project(Grant 2018YFA0704705).}}

\author{Zhe CEN\inst{1}, Xiutao FENG\inst{2}\thanks{Corresponding author: fengxt@amss.ac.cn}, Zhangyi Wang\inst{3}\and Chunping CAO\inst{1}}
\institute{ Department of Computer Science and Technology, University of Shanghai for Science and Technology, Shanghai 200093, China %
\and Key Laboratory of Mathematics Mechanization, Academy of Mathematics and Systems Sciences, CAS, Beijing 100089, China %
\and School of Cyber Science and Engineering, Wuhan University, Wuhan 430072, China
}

\begin{document}

\maketitle

\begin{abstract}
In a deduction system with some  propositions and some known relations among these propositions, people usually care about the minimum of propositions by which all other propositions can be deduced according to these known relations. Here we call it a minimizing deduction system. Its common solution is the guess and determine method. In this paper we propose a method of solving the minimizing deduction system based on MILP. Firstly, we introduce the conceptions of state variable, path variable and state copy, which enable us to characterize all rules by inequalities. Then we reduce the deduction problem to a MILP problem and solve it by the Gurobi optimizer. As its applications, we analyze the security of two stream ciphers SNOW2.0 and Enocoro-128v2 in resistance to guess and determine attacks. For SNOW 2.0, it is surprising that it takes less than 0.1s to get the best solution of 9 known variables in a personal Macbook Air(Early 2015,  Double Intel Core i5 1.6GHZ, 4GB DDR3). For Enocoro-128v2, we get the best solution of 18 known variables within 3 minutes. What's more, we propose two improvements to reduce the number of variables and inequalities which significantly decrease the scale of the MILP problem.
\end{abstract}

\keywords minimizing deduction system, guess and determine method, MILP, SNOW 2.0, Enocoro-128v2

\setcounter{secnumdepth}{4}

\section{Introduction}

In scientific researches, we often deduce other propositions by some propositions and theirs relationships. Sometimes, these propositions can be derived from each other and the relationships among them are complicated. In this situation, people are more focused on the minimum of the number of propositions which can be used to deduce all other propositions. The solution of this problem can be applied to prove some theorems in Mathematics and other fields, especially in cryptography.

At present, a common solution of such a problem is the guess and determine method. Its basic idea is to assume some propositions are viewed as axioms and gradually deduce other propositions. If they can deduce all other propositions, it is a solution. Otherwise, another propositions will be viewed as axioms  and repeat it. In cryptography, the idea of guess and determine attack firstly appeared in \cite{DIVIDE} which proposed a divide and conquer attack recovering the unknown initial state from a known keystream sequence. Golic \cite{F} applied the guess and determine attack to the alleged A5/1 and broke it theoretically. Knudsen\cite{RC4} et al. utilized it to analyze the security of RC4.  In\cite{WORD}, Hawkes and Rose extended the guess and determine attack from bits to words and gave a guess and determine attack against word-oriented stream ciphers. In\cite{SNOW1.0}, Ekdahl et al. gave a guess and determine attack anainst SNOW. Canniere \cite{GD-SOBER} presented a guess and determine attack on SOBER \cite{SOBER}. In \cite{HEURISTIC}, Ahmadi et al. proposed a heuristic guess and determine attack on stream cipher by means of some new rules derived form original rules. In\cite{PRUNE}, based on local pruning and global pruning, Charles et al. proposed a guess and determine attack on the round-reduced AES. In\cite{LFSR}, Enes Pasalic proposed a guess and determine method for filter generator. In\cite{FILTER}, Wei et al. further improved the method proposed by Enes Pasalic. In\cite{SOSEMANUK, Rabbit}, Feng et al. splitted the original word units into smaller byte units and presented a byte-based guess and determine attack to SOSEMANUK and Rabbit. Combining the idea of the guess and determine method and the time-memory tradeoff method, they further presented realtime key or state recovering attacks against a series of ciphers including A2U2\cite{A2U2}, FASER128/FASER256\cite{FASER}, Sablier\cite{Sablier} and PANDA-s\cite{PANDA-S}. In\cite{SATGD}, Oleg Zaikin et al. adopted the idea of the guess and determine method to simplify the system of equations and solved it by the SAT optimizer. 

Mixed-integer linear programming (MILP, in short) is a method to solve a mathematical optimization problem in which some or all variables are integers in order to get the minimum or maximum of an objective function. It has been wildly used in business and economics. It was introduced to compute the number of active S-boxes in differential and linear cryptanalysis by Mouha et al.\cite{FIRST} and Wu et al. \cite{BCS} respectively. Since then, MILP began to appear in cryptanalysis frequently and became a powerful automatic search tool. In \cite{S-bP}, Sun et al. extended Mouha et al.’s method for block ciphers with S-bP structure by introducing new representations for exclusive-or (XOR) differences to describe bit/word level differences simultaneously and by taking the collaborative diffusion effect of S-boxes and bitwise permutations into account. In \cite{Characteristic, ASE}, they further presented a MILP-based  automatic method for finding high probability (related-key) differential or linear characteristics of block ciphers. In \cite{ARX}, Fu et al. extended the tool of MILP to ARX ciphers. In \cite{INPOSSIBLE}, Cui et al. proposed a new automatic search tool for impossible differentials and zero-correlation linear approximations. Revealing structural properties of several ciphers from design and cryptanalysis aspects, Yu Sasaki and Yosuke Todo gave a new impossible differential search tool in \cite{INPOSSIBLE2}. Recently, Shi et al. broke the full-round MORUS by means of the MILP tool \cite{MORUS}. 

In this work we recall the conception of a minimizing deduction system and propose a novel method of solving it based on MILP solver. Firstly, in order to characterize all rules by inequalities, we introduce the conceptions of state variable, path variable and state copy. Then we reduce the deduction problem to a MILP problem and solve it by the Gurobi optimizer. As its applications, we analyze the security of two stream ciphers SNOW 2.0\cite{SNOW2.0} and Enocoro-128v2\cite{Enocoro-128v2-1} in resistance to guess and determine attacks. To our surprise, it takes less than 0.01s to get the best solution that only 9 known variables can be used to deduce all other variables for SNOW 2.0 in a personal Macbook Air(11-inch, Early 2015,  Double Intel Core i5 1.6GHZ, 4GB DDR3). For Enocoro-128v2, we get the best solution within 3 minutes that only 18 konwn variables can be used to deduce all other variables. What's more, we propose two improvements to reduce the number of variables and inequalities which significantly decrease the scale of the MILP problem.

The rest of this paper is organized as follows: in section 2, we briefly recall some preliminaries about MILP and the minimizing deduction system; in section 3, some conceptions of state variable, path variable and state copy are introduced, which enable us to characterize rules by inequalities; in section 4, as its applications, we give guess and determine attacks against two stream ciphers SNOW $2.0$ and Enocoro-128v2; in section 5, we further propose two improvements to reduce the number of variables and inequalities.

\section{Preliminaries}

\subsection{Description of a minimizing deduction system}

A deduction system is usually represented as a 2-tuple $(\mathcal{P, R})$, where $\mathcal{P}$ is a set of some propositions and $\mathcal{R}$ is a set of some relations among these propositions. In the rest of this paper, we call these relations as rules. A rule tells us what propositions a proposition can be derived from. Sometimes, a proposition may be derived from serval rules in a deduction system. For example, let $(\mathcal{P, R})$ be a deduction system of 4 propositions and 5 rules. Without loss of generality, denote $\mathcal{P} =(p_1, p_2, p_3, p_4)$ and $\mathcal{R} = (r_1, r_2, r_3, r_4, r_5 )$, where $p_1,p_2,p_3,p_4$ are propositions, and $r_1,r_2,r_3,r_4,r_5$ are rules:
\[
	\begin{split}
		r_1 :& \quad p_2 \Rightarrow p_1, \\
		r_2 :& \quad p_3, p_4 \Rightarrow p_1, \\ 
		r_3 :& \quad p_1, p_3 \Rightarrow p_2, \\
		r_4 :& \quad p_1, p_4 \Rightarrow p_3, \\
		r_5 :& \quad p_1, p_2 \Rightarrow p_4. 
	\end{split}
\]
In the above deduction system, the proposition $p_1$ can be derived from two rules $r_1$ and $r_2$, that is, $p_1$ is derived by $p_2$ in $r_1$, and by $p_3$ and $p_4$ in $r_2$.
Most of the time, people concern how many propositions, especially the minimum of the number of propositions, can deduce all other propositions in a deduction system. In this paper we call it a minimizing deduction system.  As for the above deduction system, it is easy to check the minimum of propositions which can be used to deduce all other propositions is 1.
Indeed we can deduce all other propositions by the proposition $p_2$ as below:

\begin{figure}[ht]
\centering
\includegraphics[scale=0.2]{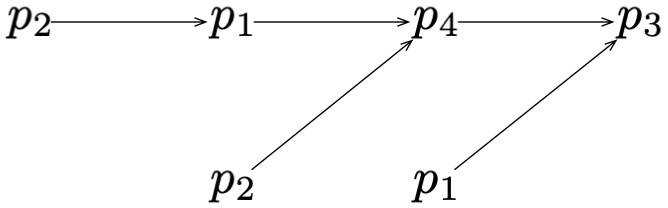}
\caption{A deduction course by the proposition $p_2$}
\label{fig:lable}
\end{figure}

\subsection{MILP}

The MILP problem is a mathematical optimization problem in which some or all variables are limited to integers. It contains three parts: an objective function, constraint conditions and decision variables. A MILP problem\cite{MILP} generally can be formalized as follows:

\begin{equation*}
\begin{split}
\textup{\textbf{Max.}}\textup{ or }\textup{\textbf{Min.}}\quad & \textbf{c}^T\textbf{x}\\
{\bf s.t.}\quad & \textbf{Ax}\le\textbf{b}\\
& \textbf{x} = \left[
 \begin{matrix}
   \textbf{p} \\
   \textbf{q}
  \end{matrix}
  \right], \quad \textbf{p}\in\mathbb{Z}^k,\quad \textbf{q}\in\mathbb{R}^{(n-k)},
\end{split}
\end{equation*}
where $\mathbb{Z}$ and $\mathbb{R}$ are denoted the set of all integers and reals respectively, $n$, $m$ and $k$ are three positive integers, $\mathbf{c}$ is a column vector in the $n$-dimensional vector space $\mathbb{R}^n$, \textbf{x} is a column vector in $\mathbb{R}^n$ in which $k$ variables are limited to integers, \textbf{A} is a full-rank metrix in $\mathbb{R}^{m \times n}$,  \textbf{b} is a column vector in $\mathbb{R}^m$. For a specific MILP problem, we can solve it by means of some mathematical softwares such as Gurobi\cite{Gurobi}, Cplex\cite{Cplex} and MiniSat\cite{MiniSat}.

\section{The characterization of inequality}

\subsection{State variable, path variable and state copy}

For a given deduction system $(\mathcal{P, R})$, where $\mathcal{P} =(p_1, p_2,\cdots, p_n)$ and $\mathcal{R} = (r_1, r_2, \cdots, r_m)$, here we will consider how to characterize it by inequalities. Firstly, we assume that the deduction system can be deduced by $k$ propositions, where $k$ is an integer such that $1\le k\le n$. For convenience, we describe a proposition $p_i(1 \le i \le n)$ as a {\it state variable} $x_i$. If a proposition $p_i$ belongs to these $k$ propositions or those deduced by these $k$ propositions, we call its corresponding state variable $x_i$  is known, denoted by 1. Otherwise, we call it unknown, denoted by 0. Sometimes a proposition $p_i$ maybe deduced by several rules $r_{i_1}, r_{i_2}, \cdots, r_{i_\tau}$, where $\tau$ is an integer such that $\tau\ge2$, $1\le i_j\le m$, $1\le j\le \tau$. In this case we call each rule $r_{i_j}$ as a path of $p_i$, denoted by a {\it path variable} $l_{i,j}$, $1\le j\le \tau$. If the proposition $p_i$ can be deduced by the rule $r_{i_j}$, we call the path variable $l_{i,j}$ is known, denoted by 1. Otherwise, we call it unknown, denoted by 0.

In order to depict the deduction system by the MILP method, we introduce a concept of state copy. Let $X=(x_1,x_2,\cdots,x_n)$ be a state of the deduction system $(\mathcal{P,R})$. We call $X'=(x'_1,x'_2,\cdots,x'_n)$ a  {\it state copy} of the state $X$ if $x'_i\ge x_i$ for all $1\le i\le n$. It is easy to see that if some state variable $x_i$ is known, then $x'_i$ is also known after the state copy. When $x_i$ is unknown, $x'_i$ maybe become known if the proposition $p_i$ can be deduced from the known propositions indicated in $X$. In the next two sections we will simulate the course of deduction by the state copy.

\subsection{The inequalities of the state variables}

For a state variable $x$, denote by $l_1,l_2,\cdots,l_\tau$ its corresponding path variables. Then $x$ is determined uniquely by  $l_1,l_2,\cdots,l_\tau$. It is easy to see that $x$ is known if and only if at least one of $l_1,l_2,\cdots,l_\tau$ is known. If all $l_i$'s are unknown, $x$ will be also unknown. The relationship of the state variable $x$ and its path variables  $l_1, l_2, \cdots, l_\tau$ is shown in Table \ref{table:1}.

\begin{table} 
\caption{The relationship of $x, l_1, l_2, \cdots, l_{\tau}$.} \label{table:1}
\centering
\setlength{\tabcolsep}{5mm}{%
\begin{tabular}{|c|c|c|c|}
\hline
No. & $x$ & $ l_1, l_2, \cdots, l_{\tau}$                            & \multicolumn{1}{l|}{Permission} \\ \hline
1  & 1 & Not all 0                      &  \checkmark                           \\ \hline
2  & 0 & All 0                          & \checkmark                            \\ \hline
3  & 0 & Not all 0 & $\times$                          \\ \hline
4  & 1 & All 0                          &  $\times$                       \\ \hline
\end{tabular}%
}
\label{Table:lable}
\end{table}

Below we consider to characterize the relationship in Table \ref{table:1} by inequalities of the form
\begin{equation} \label{eq:vpi}
	ax + \sum_{i = 1}^{\tau}b_il_i + c \ge 0.
\end{equation}
For simplification, here we assume that all $b_i(i = 1, 2, \cdots, \tau)$ are equal to $b$. According to the conditions 1, 2 and 3 in Table \ref{table:1}, we get the following inequality group:

\begin{equation} \label{eq:2}
  \left\{
   \begin{aligned}
   a + bl^{'} + c \ge 0  \\
   c \ge 0 \\
   bl^{''} + c < 0 \\
   \end{aligned}
   \right.
\end{equation}
for all $1\le l^{'}, l^{''}\le \tau$. It is easy to check that $a > 0, b < 0$ and $c \ge 0$. Thus the inequality group (\ref{eq:2}) is equivalent to the following inequality group in which $l^{'}$ and $l^{''}$ are taken $\tau$ and 1 respectively:

\begin{equation} \label{eq:3}
  \left\{
   \begin{aligned}
   a + b\tau + c \ge 0  \\
   c \ge 0 \\
   b + c < 0 \\
   \end{aligned}
   \right.
\end{equation}

Apparently, $a = \tau, b = -1, c = 0$ is a solution of the inequality group (\ref{eq:3}). Thus the conditions $1$, $2$ and $3$ can be characterized by the following inequality:

\begin{equation}\label{eq:4}
 \tau x -\sum_{i=1}^\tau l_i  \ge 0. \\
 \end{equation}
 
 Similarly, we have the inequality for the conditions $1, 2$ and $4$ in Table \ref{table:1}:
 \begin{equation}
 -2x + \sum_{i=1}^\tau l_i + 1 \ge 0. \\
 \end{equation}
 
So we get the following theorem which characterizes all conditions in Table \ref{table:1} completely.

\begin{theorem}\label{th:1}
Let $x$ be a state variable and $ l_1, l_2, \cdots, l_\tau \left(\tau\ge1\right)$ be its corresponding $\tau$ path variables. Then their relationship can be characterized by the following inequality group:
\begin{equation}
  \left\{
   \begin{aligned}
   -2x + \sum_{i=1}^\tau l_i + 1 \ge 0  \\
   \tau x - \sum_{i=1}^\tau l_i  \ge 0 \\
   \end{aligned}
   \right.
\end{equation}
\end{theorem}

\begin{proof}
The conclusion directly follows by the above deducing procedure. \hfill $\blacksquare$
\end{proof}

\noindent {\bf Remark 1: } When $\tau = 1$, we can characterize the relationship between $x$ and $l_1$ by one simpler equality $x = l_1$.

\subsection{The inequalities of the path variables}

For a given rule $r$
\[
	x_1, x_2, \cdots, x_\kappa \Rightarrow x,
\]
we introduce a path variable $l$ for $x$. Then $l$ is determined uniquely by the state variables $x_1,x_2,\cdots,x_\kappa$. Indeed we have $l$ is known if and only if all $x_i$'s are known. If at least one of $x_i$'s is unknown, then $l$ is unknown. Table \ref{table:2} shows the relationship of the path variable $l$ and its state variables  $x_1, x_2, \cdots, x_\kappa$. As for their relationship, we have the following theorem.

\begin{table}
\caption{The relationship of $l, x_1, x_2, \cdots, x_\kappa$.} \label{table:2}
\centering
\setlength{\tabcolsep}{5mm}{%
\begin{tabular}{|c|c|c|c|}
\hline
No & $l $ & {$x_1, x_2, \cdots, x_\kappa$}                             & \multicolumn{1}{l|}{Permission} \\ \hline
1  & 1 & All 1                      & \checkmark                           \\ \hline
2  & 0 & Not all 1                       & \checkmark                           \\ \hline
3  & 0 & All 1            & $\times$                          \\ \hline
4  & 1 & Not all 1                          & $\times$                          \\ \hline
\end{tabular}%
}
\label{Table:lable}
\end{table}

\begin{theorem}\label{th:2}
Let $l$ be a path variable and $ x_1, x_2, \cdots, x_\kappa (\kappa\ge1)$ be its corresponding $\kappa$ state variables. Then their relationship can be characterized by the following inequality group:

\begin{equation}\label{eq:7}
  \left\{
   \begin{aligned}
   l - \sum_{i=1}^\kappa x_i + ( \kappa  - 1 )  \ge 0 \\
   -\kappa l + \sum_{i=1}^\kappa x_i \ge 0  \\
   \end{aligned}
   \right.
\end{equation}
\end{theorem}

\begin{proof}
Let
 \begin{equation*}
   \begin{aligned}
   U &= l - \sum_{i=1}^\kappa x_i + (\kappa  - 1),   \\
   V &= -\kappa l + \sum_{i=1}^\kappa x_i.\\
   \end{aligned}
\end{equation*}

When $l = 1$ and all $x_1, x_2, \cdots, x_\kappa$ are $1$ or $l = 0$ and not all $x_1, x_2, \cdots, x_\kappa $ are 1, the value of $U$ and $V$ are always not less than $0$. That means the inequality group (\ref{eq:7}) meets the conditions $1$ and $2$ of Table \ref{table:2}. When $l = 0$ and all $x_1, x_2, \cdots, x_\kappa $ are $1$, the value of $U$ is $-1$. When  $l = 1$ and not all $x_1, x_2, \cdots, x_\kappa $ are $1$, the maximum of $V$ is $-1$. Therefore, all possible $l, x_1, x_2, \cdots, x_\kappa$ satisfying the conditions $3$ and $4$ in Table \ref{table:2} do not meet the inequality group (\ref{eq:7}). This completes the proof.  \hfill $\blacksquare$

\end{proof}

\noindent {\bf Remark 2: } When $\kappa = 1$, we can characterize the relationship between $l$ and $x_1$ by one simpler equality $l = x_1$.

\subsection{Initial condition and objective function}

For a given deduction system ($\mathcal{P, R}$) of $n$ propositions and $m$ rules, we assume that at most $k$ propositions are viewed as axioms. Let $X_0 = (x_1^{(0)}, x_2^{(0)}, \cdots, x_n^{(0)})$ be an initial state of the deduction system ($\mathcal{P, R}$). Then we have the initial condition:

\begin{equation}
x_1^{(0)} + x_2^{(0)} + \cdots + x_n^{(0)} \le k
\end{equation}

In order to simulate the deduction course, we conduct the state copy. Let $X_1, X_2, \cdots, X_\nu$ be the state sequence, where $X_i=(x_1^{(i)}, x_2^{(i)}, \cdots, x_n^{(i)})$, $1\le i\le \nu$, $\nu$ is a positive integer, and $X_i$ is the state copy of $X_{i-1}$. If a state variable $x$ in $X_{i-1}$ is known, then it is also known in $X_{i}$. If it is unknown in $X_{i-1}$ but it can be deduced by other known state variables in $X_{i-1}$ according to the rules $\mathcal{R}$, then it is known in $X_i$. After the state copy $\nu$ times, we expect the 1's in $X_\nu$ as many as possible, that is:
\begin{equation}
\textup{\textbf{Max.}}\quad x_1^{(\nu)} + x_2^{(\nu)} + \cdots + x_n^{(\nu)}.
\end{equation}
Obviously, $X_0 = (x_1^{(0)}, x_2^{(0)}, \cdots, x_n^{(0)})$ is a solution of the deduction system ($\mathcal{P, R}$) when $x_1^{(\nu)} + x_2^{(\nu)} + \cdots + x_n^{(\nu)}=n$.

\section{Applications to stream ciphers}

\subsection{SNOW 2.0} 
SNOW 2.0~\cite{SNOW2.0} is a strengthened version of SNOW\cite{SNOW1.0} which is a candidate of the NESSIE  and has been broken by the guess and determine attack\cite{SNOWGD} and the distinguishing attack\cite{SNOWDis}. At present SNOW 2.0 become one of the most important stream ciphers and is selected to be an international standard by ISO/IEC 18033-4\cite{Standard}. Below we recall SNOW 2.0 briefly. Please refer to~\cite{SNOW2.0} for more details.  

\subsubsection{Description of SNOW 2.0} $\\$

\noindent The stream cipher SNOW 2.0 consists of a Linear Feedback Shift Register $\left( LFSR \right) $ and a Finite State Machine $\left( FSM \right) $ as shown in Fig. \ref{Fig:snow2}. Here we denote by $\oplus$ and $\boxplus$ the bitwise addition and  the integer addition modulo $2^{32}$ respectively, and by $\alpha$ a root of the polynomial $x^4 + \beta^{23}x^3 + \beta^{245}x^2 + \beta^{48}x + \beta^{239} \in\mathbb{F}_{2^8}\left[x\right]$, where $\beta$ is a root of the polynomial $x^8 + x^7 + x^5 + x^3 + 1\in\mathbb{F}_2\left[x\right]$. Let $\alpha^{-1}$ be the inverse of $\alpha$, and $S$ be a permutation on $\mathbb{Z}_{2^{32}}$ based on the round function of Rijndael\cite{Rijndael}. We have the following iterative relations:

\begin{equation} \label{eq:snow1}
  \left\{
   \begin{aligned}
   &s_{t+16} = \alpha^{-1} \cdot s_{t+11} \oplus  s_{t+2} \oplus \alpha \cdot s_t \\
   &z_{t} = s_{t+15}\boxplus R1_t \oplus R2_t \oplus s_t \\
   &R1_{t+1} = s_{t+5} \boxplus R2_t\\
   &R2_{t+1} = S(R1_t)\\
   \end{aligned}
   \right.
\end{equation}
where $t \ge 0, t\in \mathbb{Z}$, and $z_{t}$ is the output key word.  We always assume that the sequence $\{z_t\}_{t\ge0}$ is known in the context of a guess and determine attack.

\begin{figure}[ht]
\centering
\includegraphics[scale=0.04]{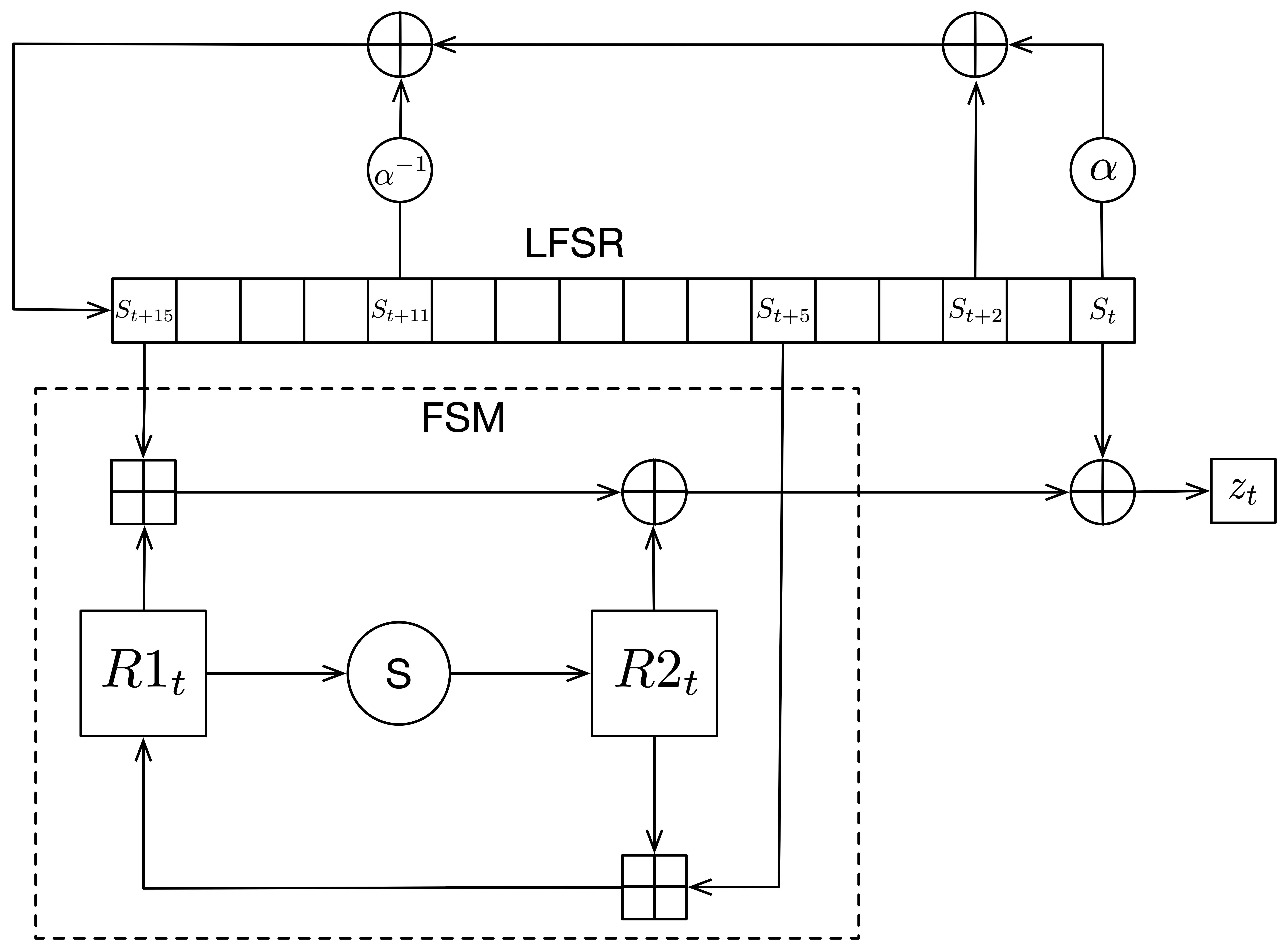}
\caption{A schematic picture of SNOW 2.0} \label{Fig:snow2}
\end{figure}

\subsubsection{Inequality characterization of SNOW 2.0}$\\$

\noindent In our minimizing deduction system, since we mainly focus on whether a variable is known or not rather than its certain value, we indicate by a state variable whether the value of some register unit is known or not. Without confusion, we still adopt the same symbol to indicate the state variable of some register unit, for example, the state variable $s_t$ is used to indicate the state of the LFSR unit $s_t$.  If the value of the LFSR unit $s_t$ is known, then set the state variable $s_t=1$, otherwise, set $s_t=0$. It is noticed that all variables in every formula of the iterative relations (\ref{eq:snow1}) are symmetrical, that is, each variable can be deduced by the other variables. Below we denote a rule by $\left[x_1, x_2, \cdots, x_\kappa\right]$, where $x_i$ is a state variable, $1 \le i \le \kappa$, which means each $x_i$ in $\left[x_1, x_2, \cdots, x_\kappa\right]$ can be deduced by the other $\left(\kappa-1\right)$ state variables and is called a {\it symmetrical rule}. Therefore we can get all rules by the iterative relations (\ref{eq:snow1}) as follows:

\begin{equation}\label{eq:snow2}
  \left\{
   \begin{aligned}
   &[s_{t+16}, s_{t+11}, s_{t+2}, s_t] \\
   &[s_{t+15}, R1_t, R2_t, s_t] \\
   &[R1_{t+1}, s_{t+5}, R2_t]\\
   &[R2_{t+1}, R1_t]\\
   \end{aligned}
   \right.
\end{equation}
where $t\ge0, t\in \mathbb{Z}$. Note that $[R2_{t+1}, R1_t]$ means one is known if and only if the other is also known, we always have $R2_{t+1}=R1_t$ for an arbitrary $t\ge0$. Therefore we can eliminate some state variables to further simplify the rules (\ref{eq:snow2}) and get the following rules:

\begin{equation}
  \left\{
   \begin{aligned}
   &[s_{t+16}, s_{t+11}, s_{t+2}, s_t] \\
   &[s_{t+15}, R2_{t+1}, R2_t, s_t] \\
   &[R2_{t+2}, s_{t+5}, R2_t]\\
   \end{aligned}
   \right.
\end{equation}

Rewriting R2 as R, we can get the following rules:

\begin{subequations}
\renewcommand{\theequation}
{\theparentequation-\arabic{equation}}
\begin{align}
&[s_{t+16}, s_{t+11}, s_{t+2}, s_t] \\
&[s_{t+15}, R_{t+1}, R_t, s_t] \\
&[R_{t+2}, s_{t+5}, R_t]
\end{align}
\end{subequations}

Assuming that totally $T$ key words $z_0, z_1, \cdots, z_{T-1}$ are known in a guess and determine attack, where $T\ge1$. Thus we get the state $X$ of all state variables in SNOW 2.0:
\[
	X = (s_0, s_1, s_2,\cdots,s_{14+T}, R_0, R_1, \cdots, R_{T}),
\]
which contains $(2T+16)$ state variables. We repeat the state copy $\nu$ times and get a state sequence $\{X_i\}_{0\le i\le \nu}$:
\[
	X^{(i)} = (s_0^{(i)}, s_1^{(i)}, s_2^{(i)},\cdots,s_{14+T}^{(i)}, R_0^{(i)}, R_1^{(i)}, \cdots, R_{T}^{(i)}).
\]

For any $0\le i\le\nu-1$, below we consider how many rules each state variable in $X_{i+1}$ can be deduced from $X_{i}$ by. Take the state variable $s_{11}^{(i+1)}$ for example, where we assume that $T\ge13$. By the state copy of $X_i$ and the rules (13-1), (13-2) and (13-3), we get all 6 rules deducing $s_{11}^{(i+1)}$:
\[
	\begin{split}
	 & \{s_{11}^{(i)}\}, \\
	 &\{R_{6}^{(i)}, R_{8}^{(i)}\},  \\
	 & \{s_{13}^{(i)}, s_{22}^{(i)}, s_{27}^{(i)}\},  \\
	 & \{R_{11}^{(i)}, R_{12}^{(i)}, s_{26}^{(i)}\},  \\
	 & \{s_{9}^{(i)}, s_{20}^{(i)}, s_{25}^{(i)}\},  \\
	 &\{s_{0}^{(i)}, s_{2}^{(i)}, s_{16}^{(i)}\}.
	\end{split}
\]

For the above each rule, we introduce a path variable $s_{11, j}^{(i+1)}$, $1 \le j \le 6$. By Theorems \ref{th:1} and \ref{th:2}, we get a characterization on the update of $s_{11}$.

\begin{equation}
  \left\{
   \begin{aligned}
   s_{11, 1}^{(i+1)} - s_{11}^{(i)} = 0 \\
   s_{11, 5}^{(i+1)} - R_{6}^{(i)} - R_{8}^{(i)} + 1 \ge 0  \\
   -2s_{11, 5}^{(i+1)} + R_{6}^{(i)} + R_{8}^{(i)}  \ge 0 \\
   s_{11, 2}^{(i+1)} -  s_{13}^{(i)} - s_{22}^{(i)} - s_{27}^{(i)} + 2 \ge 0  \\
   -3s_{11, 2}^{(i+1)} + s_{13}^{(i)} + s_{22}^{(i)} + s_{27}^{(i)}  \ge 0 \\
   s_{11, 3}^{(i+1)} - R_{11}^{(i)} - R_{12}^{(i)} - s_{26}^{(i)} + 2 \ge 0  \\
   -3s_{11, 3}^{(i+1)} + R_{11}^{(i)} + R_{12}^{(i)} + s_{26}^{(i)}  \ge 0 \\
   s_{11, 4}^{(i+1)} - s_{9}^{(i)} - s_{20}^{(i)} - s_{25}^{(i)} + 2 \ge 0  \\
   -3s_{11, 4}^{(i+1)} + s_{9}^{(i)} + s_{20}^{(i)} + s_{25}^{(i)}  \ge 0 \\
   s_{11, 6}^{(i+1)} - s_{0}^{(i)} - s_{2}^{(i)} - s_{16}^{(i)} + 2 \ge 0  \\
   -3s_{11, 6}^{(i+1)} + s_{0}^{(i)} + s_{2}^{(i)} + s_{16}^{(i)}  \ge 0 \\
   -2s_{11}^{(i+1)} + \sum_{j=1}^{6}s_{11, j}^{(i+1)} + 1\ge 0  \\
   6s_{11}^{(i+1)} - \sum_{j=1}^{6}s_{11, j}^{(i+1)} \ge 0
   \end{aligned}
   \right.
\end{equation}

For the other state variables in $X_{i+1}$, please refer to Appendix A. 

Here we assume that at most $k$ state variables in the initial state $X_0$ are known. Then we have the initial constraint condition
\begin{equation}
   \sum_{i=0}^{14+T}s_{i}^{(0)} + \sum_{i=0}^{T}R_{i}^{(0)} \le k, 
\end{equation}
and the objective function
\begin{equation}
\begin{split}
{\bf Max.}\quad & \sum_{i=0}^{14+T}s_{i}^{(\nu)} + \sum_{i=0}^{T}R_{i}^{(\nu)}. 
\end{split}
\end{equation}

\subsubsection{Experimental result}$\\$

\noindent In our experiment, we take $T = 13, \nu = 12$ and $k = 9$ and solve it by the Gurobi optimizer. It is surprising that it takes less than 0.1s to get the best solution of $9$ known variables: $R_4, R_5, R_6, R_7, R_8, R_9, R_{10}, R_{11}, R_{12}$. The deduction course of all other variables are given in Table \ref{la:snow_course}  according to the result of the Gurobi optimizer. 

\begin{table}[] 
\caption{The deduction course of state variables in SNOW 2.0}\label{la:snow_course}
\centering
\setlength{\tabcolsep}{1.5mm}{ 
\begin{tabular}{|c|c|l|c||c|c|l|c|} 
\hline
No. & \multicolumn{1}{c|}{Known}&Rule  & Deduced & No.  & \multicolumn{1}{c|}{Known} & Rule & Deduced\\ \hline\hline
1 & $R_{4},R_{6}$ & 13-3 & $s_{9}$  & 18 & $s_{10},s_{19},s_{24}$ & 13-1 & $s_{8}$  \\ \hline
2 & $R_{5},R_{7}$ & 13-3 & $s_{10}$  & 19 & $s_{8},R_{8},R_{9}$ & 13-2 & $s_{23}$  \\ \hline
3 & $R_{6},R_{8}$ & 13-3 & $s_{11}$  & 20 & $s_{8},R_{5}$ & 13-3 & $R_{3}$  \\ \hline
4 & $R_{7},R_{9}$ & 13-3 & $s_{12}$  & 21 & $s_{3},R_{3},R_{4}$ & 13-2 & $s_{18}$  \\ \hline
5 & $R_{8},R_{10}$ & 13-3 & $s_{13}$  & 22 & $s_{6},R_{3}$ & 13-3 & $R_{1}$  \\ \hline
6 & $R_{9},R_{11}$ & 13-3 & $s_{14}$  & 23 & $s_{4},s_{13},s_{18}$ & 13-1 & $s_{2}$  \\ \hline
7 & $R_{10},R_{12}$ & 13-3 & $s_{15}$  & 24 & $s_{9},s_{18},s_{23}$ & 13-1 & $s_{7}$  \\ \hline
8 & $s_{9},R_{9},R_{10}$ & 13-2 & $s_{24}$  & 25 & $s_{5},s_{7},s_{21}$ & 13-1 & $s_{16}$  \\ \hline
9 & $s_{10},R_{10},R_{11}$ & 13-2 & $s_{25}$  & 26 & $s_{7},R_{7},R_{8}$ & 13-2 & $s_{22}$  \\ \hline
10 & $s_{11},R_{11},R_{12}$ & 13-2 & $s_{26}$  & 27 & $s_{7},R_{4}$ & 13-3 & $R_{2}$  \\ \hline
11 & $s_{9},s_{11},s_{25}$ & 13-1 & $s_{20}$  & 28 & $s_{2},s_{11},s_{16}$ & 13-1 & $s_{0}$  \\ \hline
12 & $s_{10},s_{12},s_{26}$ & 13-1 & $s_{21}$  & 29 & $s_{6},s_{8},s_{22}$ & 13-1 & $s_{17}$  \\ \hline
13 & $s_{20},R_{5},R_{6}$ & 13-2 & $s_{5}$  & 30 & $s_{11},s_{13},s_{22}$ & 13-1 & $s_{27}$  \\ \hline
14 & $s_{21},R_{6},R_{7}$ & 13-2 & $s_{6}$  & 31 & $s_{16},R_{1},R_{2}$ & 13-2 & $s_{1}$  \\ \hline
15 & $s_{6},s_{15},s_{20}$ & 13-1 & $s_{4}$  & 32 & $s_{5},R_{2}$ & 13-3 & $R_{0}$  \\ \hline
16 & $s_{4},R_{4},R_{5}$ & 13-2 & $s_{19}$  & 33 & $s_{16},R_{11}$ & 13-3 & $R_{13}$  \\ \hline
17 & $s_{5},s_{14},s_{19}$ & 13-1 & $s_{3}$  &      &                      &         &  \\ \hline
\end{tabular}
}
\end{table}

\subsection{Enocoro-128v2}
Enocoro is a family of stream ciphers whose common specification firstly was published in\cite{Enocoro1.0}. We mainly focus on Enocoro-128v2\cite{Enocoro-128v2-1} in this paper. Sine the update function of Enocoro-128v2 is the same as Enocoro-128v1.1\cite{Enocoro-128v1.1}, our method have the same effect on both. Below we only give a short description of Enocoro-128v2. More details are in \cite{Enocoro-128v2-1, Enocoro-128v2-2}.

\subsubsection{Description of Enocoro-128v2}  $\\$

\noindent Enocoro-128v2 consists of four LFSR and a FSM as shown in Fig.\ref{lb:orginal_enocoro}. The substitution box $s_8$ defines a permutation which maps 8-bits inputs to 8-bits outputs. Here we don't care about the inner details of $s_8$.  The linear transformation $L$ is defined by a $2\times 2$ matrix over GF$\left(2^8\right)$:

\begin{equation}\label{eq:L}
    \begin{pmatrix}
        v_0 \\
        v_1 \\
    \end{pmatrix} = L(u_0, u_1) =  
    \begin{pmatrix}
        1 & 1  \\
        1 & 2\\
    \end{pmatrix}
    \begin{pmatrix}
        u_0 \\
        u_1 \\
    \end{pmatrix}.
 \end{equation}
According to the above expression, it is known that as long as arbitrary two values of $v_0, v_1, u_0, u_1$ are known, the other two values will be calculated. For the convenience of description, we redraw Fig. \ref{lb:orginal_enocoro} as Fig. \ref{lb:revised_enocoro}. As for Enocono-128v2, we have the following  iterative relations:

\begin{figure}[ht]
\centering
\includegraphics[scale=0.0235]{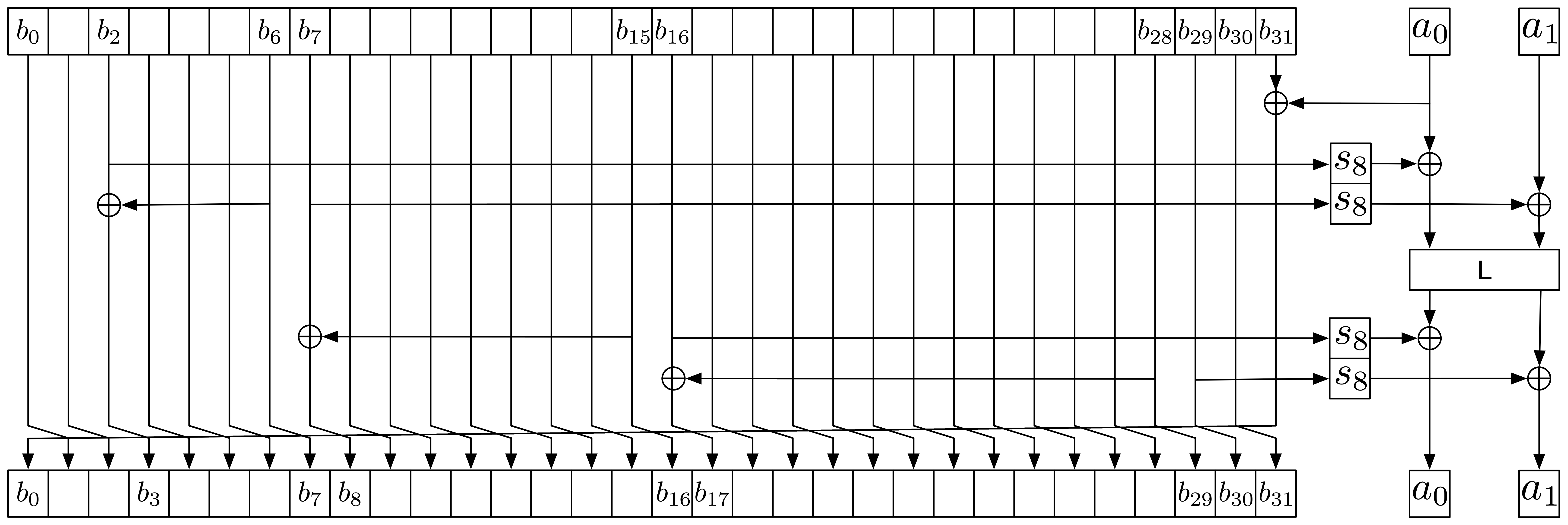}
\caption{The structure of Enocoro-128v2}\label{lb:orginal_enocoro}
\label{fig:lable}
\end{figure}

\begin{figure}[ht]
\centering
\includegraphics[scale=0.035]{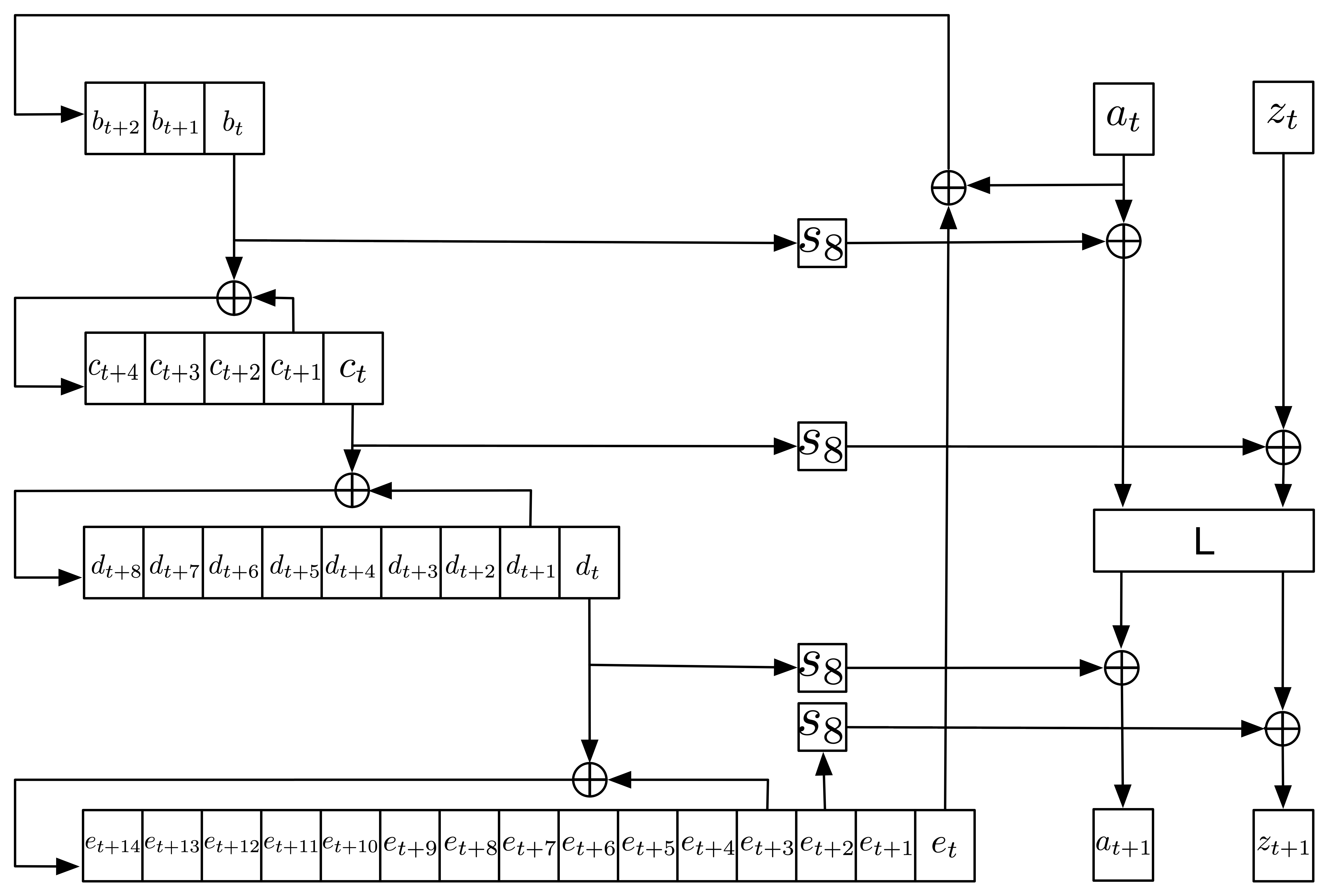}
\caption{The equivalent structure of Enocoro-128v2}\label{lb:revised_enocoro}
\label{fig:lable}
\end{figure}

\begin{equation}\label{eq:enocoro}
  \left\{
   \begin{aligned}
   &b_{t+3} = a_{t} \oplus  e_{t}\\
   &c_{t+5} = b_{t}\oplus c_{t+1} \\
   &d_{t+9} = c_{t} \oplus d_{t+1}\\
   &e_{t+15} = d_{t} \oplus e_{t+3}\\
   &\begin{pmatrix}
        a_{t+1} \oplus s_8(d_t) \\
        z_{t+1} \oplus s_8(e_{t+2}) \\
    \end{pmatrix} =  L
    \begin{pmatrix}
        a_t \oplus s_8(b_t) \\
        z_t \oplus s_8(c_t) \\
    \end{pmatrix}
   \end{aligned}
   \right.
\end{equation}
where $t \ge 0, t\in \mathbb{Z}$ and $z_{t}$ is the output key word.  We always assume that the sequence $\{z_t\}_{t\ge0}$ is known in the context of a guess and determine attack.

\subsubsection{Inequality characterization of Enocoro-128v2}$\\$

The course of translating iterative relations (\ref{eq:enocoro})  into rules  is similar to SNOW 2.0. The only different part is how to handle the operation of the linear transform $L$. Here we introduce two variables $f_t$ and $g_t$ where $f_{t} = a_{t} \oplus  s_8(b_t)$, $g_{t} = a_{t+1} \oplus  s_8(d_t)$. The iterative relations (\ref{eq:enocoro}) can be translated as the following:

\begin{equation}
  \left\{
   \begin{aligned}
   &b_{t+3} = a_{t} \oplus  e_{t}\\
   &c_{t+5} = b_{t}\oplus c_{t+1} \\
   &d_{t+9} = c_{t} \oplus d_{t+1}\\
   &e_{t+15} = d_{t} \oplus e_{t+3}\\
   &f_{t} = a_{t} \oplus  s_8(b_t)\\
   &g_{t} = a_{t+1} \oplus  s_8(d_t)\\
   &\begin{pmatrix}
        g_{t}  \\
        z_{t+1} \oplus s_8(e_{t+2}) \\
    \end{pmatrix} =  L
    \begin{pmatrix}
        f_{t} \\
        z_t \oplus s_8(c_t) \\
    \end{pmatrix}
   \end{aligned}
   \right.
\end{equation}

 Since $z_t$ and $z_{t+1}$ are known variables and $S_8$ is a known permutation, we omit them in our minimizing deduction system. For the linear transformation $L$, as long as we know two variables in this transformation, the other two variables can be deduced. Therefore we have the following rules:

\begin{subequations}
\renewcommand{\theequation}
{\theparentequation-\arabic{equation}}
\begin{align}
&[b_{t+3}, a_{t}, e_{t}] \\
&[c_{t+5}, b_{t}, c_{t+1}] \\
 &[d_{t+9}, c_{t}, d_{t+1}]\\
   &[e_{t+15}, d_{t}, e_{t+3}]\\
   &[f_{t}, a_{t}, b_{t}]\\
   &[g_{t}, a_{t+1}, d_{t}]\\
   &[g_{t}, f_{t}, e_{t+2}]\\
   &[g_{t}, f_{t}, c_{t}]\\
   &[g_{t}, e_{t+2}, c_{t}]\\
   &[f_{t}, e_{t+2}, c_{t}]
\end{align}
\end{subequations}

Assuming that totally $T$ key words $z_0, z_1, \cdots, z_{T-1}$ are known in a guess and determine attack, where $T\ge2$. Thus we get the state $X$ of all state variables in Enocoro-128v2:

\[
	\begin{split}
	X =& (a_0, \cdots,a_{T-1}, b_0,\cdots, b_{T-2}, c_0,\cdots, c_{T-2}, d_0, \cdots, d_{T-2}, \\
	 & \quad e_0, \cdots, e_{T}, f_0,\cdots, f_{T-2}, g_0, \cdots, g_{T-2}),
	 \end{split}
\]
which contains $(7T-4)$ state variables. We repeat the state copy $\nu$ times and get a state sequence $\{X_i\}_{0\le i\le \nu}$:
\[
	\begin{split}
	X^{(i)} =& (a_0^{(i)}, \cdots,a_{T-1}^{(i)}, b_0^{(i)}, \cdots, b_{T-2}^{(i)}, c_0^{(i)}, \cdots, c_{T-2},^{(i)} d_0^{(i)}, \cdots, d_{T-2}^{(i)}, \\
	 & \quad e_0^{(i)},\cdots, e_{T}^{(i)}, f_0^{(i)},\cdots, f_{T-2}^{(i)}, g_0, \cdots, g_{T-2}^{(i)}).
	 \end{split}
\]

For every state variables in $X$,  since their inequality characterization is similar to SNOW 2.0 by Theorems \ref{th:1} and \ref{th:2}, we don't repeat them. Please refer to Appendix B for more details.

Here we assume that at most $k$ state variables in the initial state $X_0$ are known. Then we have the initial constraint condition
\begin{equation}
   \sum_{i=0}^{T-1}a_{i}^{(0)} + \sum_{i=0}^{T-2}b_{i}^{(0)} + \sum_{i=0}^{T-2}c_{i}^{(0)} + \sum_{i=0}^{T-2}d_{i}^{(0)} + \sum_{i=0}^{T}e_{i}^{(0)} + \sum_{i=0}^{T-2}f_{i}^{(0)} + \sum_{i=0}^{T-2}g_{i}^{(0)} \le k,
\end{equation}
and the objective function
\begin{equation}\small
\begin{split}
{\bf Max.}\  \sum_{i=0}^{T-1}a_{i}^{(\nu)} + \sum_{i=0}^{T-2}b_{i}^{(\nu)} + \sum_{i=0}^{T-2}c_{i}^{(\nu)} + \sum_{i=0}^{T-2}d_{i}^{(\nu)} + \sum_{i=0}^{T}e_{i}^{(\nu)} + \sum_{i=0}^{T-2}f_{i}^{(\nu)} + \sum_{i=0}^{T-2}g_{i}^{(\nu)}.
\end{split}
\end{equation}

\subsubsection{Experimental result}$\\$

\noindent In our experiment, we take $T = 16, \nu = 18$ and $k = 18$ and solve it by the Gurobi optimizer. It takes about 3 minutes to get a current best solution of $18$ known variables: $a_{3}$, $a_{5}$, $b_{2}$, $b_{5}$, $b_{6}$, $c_{2}$, $c_{3}$, $c_{8}$, $c_{9}$, $c_{10}$, $e_{6}$, $e_{11}$, $e_{15}$, $f_{3}$, $f_{6}$,  $g_{1}$, $g_{2}$ $g_{5} $, whose objective function reaches 92. We find they indeed are a group of solutions by verification. 
The deduction course of all other variables are given in Table \ref{la:cnocoro_course}  according to the result of the Gurobi optimizer.

\begin{table}[]
\caption{The deduction courses of variables of Enocoro-128v2}\label{la:cnocoro_course}
\centering\small
\begin{tabular}{|c|c|c|c||c|c|c|c||c|c|c|c|}
\hline
NO. & Known & Rule & Ded.&NO. & Known & Rule & Ded.&NO. & Known & Rule & Ded.\\ \hline \hline
1 & $a_{3},b_{6}$ & 20-1 & $e_{3}$  & 34 & $c_{1},d_{2}$ & 20-3 & $d_{10}$  & 66 & $c_{10},g_{10}$ & 20-9 & $e_{12}$  \\ \hline
2 & $b_{2},c_{3}$ & 20-2 & $c_{7}$  & 35 & $c_{4},d_{5}$ & 20-3 & $d_{13}$  & 67 & $a_{10},f_{10}$ & 20-5 & $b_{10}$  \\ \hline
3 & $b_{5},c_{10}$ & 20-2 & $c_{6}$  & 36 & $b_{1},f_{1}$ & 20-5 & $a_{1}$  & 68 & $b_{10},e_{7}$ & 20-1 & $a_{7}$  \\ \hline
4 & $a_{3},f_{3}$ & 20-5 & $b_{3}$  & 37 & $b_{9},f_{9}$ & 20-5 & $a_{9}$  & 69 & $b_{10},c_{11}$ & 20-2 & $c_{15}$  \\ \hline
5 & $a_{5},b_{5}$ & 20-5 & $f_{5}$  & 38 & $a_{2},g_{1}$ & 20-6 & $d_{1}$  & 70 & $a_{7},b_{7}$ & 20-5 & $f_{7}$  \\ \hline
6 & $b_{6},f_{6}$ & 20-5 & $a_{6}$  & 39 & $c_{4},e_{6}$ & 20-9 & $g_{4}$  & 71 & $a_{7},g_{6}$ & 20-6 & $d_{6}$  \\ \hline
7 & $a_{3},g_{2}$ & 20-6 & $d_{2}$  & 40 & $c_{4},e_{6}$ & 20-10 & $f_{4}$  & 72 & $c_{15},e_{17}$ & 20-9 & $g_{15}$  \\ \hline
8 & $c_{2},g_{2}$ & 20-8 & $f_{2}$  & 41 & $a_{1},b_{4}$ & 20-1 & $e_{1}$  & 73 & $c_{15},e_{17}$ & 20-10 & $f_{15}$  \\ \hline
9 & $c_{3},f_{3}$ & 20-8 & $g_{3}$  & 42 & $d_{1},e_{4}$ & 20-4 & $e_{16}$  & 74 & $c_{5},d_{6}$ & 20-3 & $d_{14}$  \\ \hline
10 & $c_{2},g_{2}$ & 20-9 & $e_{4}$  & 43 & $b_{4},f_{4}$ & 20-5 & $a_{4}$  & 75 & $c_{7},f_{7}$ & 20-8 & $g_{7}$  \\ \hline
11 & $c_{9},e_{11}$ & 20-9 & $g_{9}$  & 44 & $a_{1},d_{0}$ & 20-6 & $g_{0}$  & 76 & $c_{7},f_{7}$ & 20-10 & $e_{9}$  \\ \hline
12 & $c_{3},f_{3}$ & 20-10 & $e_{5}$  & 45 & $a_{5},g_{4}$ & 20-6 & $d_{4}$  & 77 & $a_{9},e_{9}$ & 20-1 & $b_{12}$  \\ \hline
13 & $c_{9},e_{11}$ & 20-10 & $f_{9}$  & 46 & $c_{13},e_{15}$ & 20-9 & $g_{13}$  & 78 & $d_{14},g_{14}$ & 20-6 & $a_{15}$  \\ \hline
14 & $a_{5},e_{5}$ & 20-1 & $b_{8}$  & 47 & $c_{13},e_{15}$ & 20-10 & $f_{13}$  & 79 & $a_{15},f_{15}$ & 20-5 & $b_{15}$  \\ \hline
15 & $a_{6},e_{6}$ & 20-1 & $b_{9}$  & 48 & $a_{4},e_{4}$ & 20-1 & $b_{7}$  & 80 & $b_{15},e_{12}$ & 20-1 & $a_{12}$  \\ \hline
16 & $c_{2},c_{6}$ & 20-2 & $b_{1}$  & 49 & $c_{3},d_{4}$ & 20-3 & $d_{12}$  & 81 & $a_{12},b_{12}$ & 20-5 & $f_{12}$  \\ \hline
17 & $b_{3},c_{8}$ & 20-2 & $c_{4}$  & 50 & $a_{4},g_{3}$ & 20-6 & $d_{3}$  & 82 & $a_{12},d_{11}$ & 20-6 & $g_{11}$  \\ \hline
18 & $b_{6},c_{7}$ & 20-2 & $c_{11}$  & 51 & $d_{13},g_{13}$ & 20-6 & $a_{14}$  & 83 & $c_{11},g_{11}$ & 20-8 & $f_{11}$  \\ \hline
19 & $e_{3},e_{15}$ & 20-4 & $d_{0}$  & 52 & $e_{2},g_{0}$ & 20-7 & $f_{0}$  & 84 & $c_{12},f_{12}$ & 20-8 & $g_{12}$  \\ \hline
20 & $d_{2},e_{5}$ & 20-4 & $e_{17}$  & 53 & $e_{2},g_{0}$ & 20-9 & $c_{0}$  & 85 & $c_{11},g_{11}$ & 20-9 & $e_{13}$  \\ \hline
21 & $b_{2},f_{2}$ & 20-5 & $a_{2}$  & 54 & $c_{14},e_{16}$ & 20-9 & $g_{14}$  & 86 & $c_{12},f_{12}$ & 20-10 & $e_{14}$  \\ \hline
22 & $a_{6},g_{5}$ & 20-6 & $d_{5}$  & 55 & $c_{14},e_{16}$ & 20-10 & $f_{14}$  & 87 & $a_{11},f_{11}$ & 20-5 & $b_{11}$  \\ \hline
23 & $e_{3},g_{1}$ & 20-7 & $f_{1}$  & 56 & $b_{7},c_{8}$ & 20-2 & $c_{12}$  & 88 & $d_{12},g_{12}$ & 20-6 & $a_{13}$  \\ \hline
24 & $f_{5},g_{5}$ & 20-7 & $e_{7}$  & 57 & $c_{0},d_{1}$ & 20-3 & $d_{9}$  & 89 & $b_{11},e_{8}$ & 20-1 & $a_{8}$  \\ \hline
25 & $f_{5},g_{5}$ & 20-8 & $c_{5}$  & 58 & $c_{2},d_{3}$ & 20-3 & $d_{11}$  & 90 & $a_{13},f_{13}$ & 20-5 & $b_{13}$  \\ \hline
26 & $c_{6},f_{6}$ & 20-8 & $g_{6}$  & 59 & $b_{0},f_{0}$ & 20-5 & $a_{0}$  & 91 & $a_{10},b_{13}$ & 20-1 & $e_{10}$  \\ \hline
27 & $e_{3},g_{1}$ & 20-9 & $c_{1}$  & 60 & $a_{14},f_{14}$ & 20-5 & $b_{14}$  & 92 & $a_{8},b_{8}$ & 20-5 & $f_{8}$  \\ \hline
28 & $c_{6},f_{6}$ & 20-10 & $e_{8}$  & 61 & $a_{0},b_{3}$ & 20-1 & $e_{0}$  & 93 & $a_{8},g_{7}$ & 20-6 & $d_{7}$  \\ \hline
29 & $a_{2},b_{5}$ & 20-1 & $e_{2}$  & 62 & $b_{14},e_{11}$ & 20-1 & $a_{11}$  & 94 & $c_{6},d_{7}$ & 20-3 & $d_{15}$  \\ \hline
30 & $c_{1},c_{5}$ & 20-2 & $b_{0}$  & 63 & $d_{9},g_{9}$ & 20-6 & $a_{10}$  & 95 & $e_{10},f_{8}$ & 20-7 & $g_{8}$  \\ \hline
31 & $c_{5},c_{9}$ & 20-2 & $b_{4}$  & 64 & $a_{11},d_{10}$ & 20-6 & $g_{10}$  & 96 & $a_{9},g_{8}$ & 20-6 & $d_{8}$  \\ \hline
32 & $b_{8},c_{9}$ & 20-2 & $c_{13}$  & 65 & $c_{10},g_{10}$ & 20-8 & $f_{10}$  & 97 & $d_{15},g_{15}$ & 20-6 & $a_{16}$  \\ \hline
33 & $b_{9},c_{10}$ & 20-2 & $c_{14}$  & & & & & & &&\\ \hline
\end{tabular}
\end{table}

Our result is better than that in \cite{Enocoro-128v2_evaluation} which has to guess 20-22 variables in total and is the same as in \cite{GUESS18}. However, we don't have access to this paper since it is in Japanese.

\section{Two improvements}

\subsection{The reduction of path variables and inequalities}

In this section we will improve the minimizing deduction system by reducing the number of path variables and inequalities. In the previous description,
a state variable is determined by its corresponding path variables which consist of tow parts: one from the state copy and the others from some rules.
We find that for the update of each state variable $x$, at least 2 path variables which come from the state copy and some rule respectively can be reduced. More concretely, let $x$ be a state variable and $l_1, l_2, \cdots, l_{\tau}$ be its corresponding path variables, where $l_1$ is from the state copy. We suppose that $l_\tau$ is determined by the state variables $x_1, x_2, \cdots, x_\kappa$. Denote by $x'$ the variable copy of $x$. Then the relation between $x'$ and $x, l_2, \cdots, l_{\tau-1}, x_1, x_2, \cdots, x_\kappa$ is shown in Table \ref{tb:x'}. We have the following conclusion:

\begin{table}[]
\caption{The relationship of $x', x, l_2, \cdots, l_{\tau-1}, x_1, x_2, \cdots, x_{\kappa}$.}\label{tb:x'}
\centering
\setlength{\tabcolsep}{5mm}{%
\begin{tabular}{|c|c|c|c|c|c|}
\hline
No. & $x'$ & {$ x, l_2, \cdots, l_{\tau-1}$}        & {$x_1, x_2, \cdots, x_{\kappa}$}         & permission \\ \hline
1  & 1 & Not all 0 &  N/A         & \checkmark         \\ \hline
2  & 1 & N/A          & All 1     & \checkmark         \\ \hline
3  & 0 & All 0     & Not all 1 & \checkmark         \\ \hline
4  & 1 & All 0     & Not all 1 & $\times$     \\ \hline
5  & 0 & Not all 0 &  N/A         & $\times$     \\ \hline
6  & 0 & N/A          & All 1     & $\times$     \\ \hline
\end{tabular}%
}
\label{Table:lable}
\end{table}

\begin{theorem}\label{th:3}
Let $x'$ be a copy of the state variable $x$, $l_2, \cdots, l_{\tau}$ be path variables corresponding to $x$. Denote by $x_1, x_2, \cdots, x_{\kappa}$
the state variables of $l_\tau$. Then they are characterized by the following inequality group:
\begin{equation}\label{eq:improve1}
  \left\{
   \begin{aligned}
   ((\tau - 1)\kappa + 1)x' - \kappa(x + \sum_{i=2}^{\tau-1}l_i) -  \sum_{i=1}^{\kappa}x_i + \kappa - 1 \ge 0 \\
   -\kappa x' + \kappa (x + \sum_{i=2}^{\tau-1}l_i) + \sum_{i=1}^{\kappa}x_i  \ge 0 \\
   \end{aligned}
   \right.
\end{equation}
\end{theorem}

\begin{proof}

Similarly to the proof of theorem $\ref{th:2}$, let
 \begin{equation*}
   \begin{aligned}
   U  &= ((\tau - 1)\kappa + 1)x' - \kappa(x + \sum_{i=2}^{\tau-1}l_i) -  \sum_{i=1}^{\kappa}x_i + \kappa - 1, \\
    V &= -\kappa x' + \kappa (x + \sum_{i=2}^{\tau-1}l_i) + \sum_{i=1}^{\kappa}x_i.
   \end{aligned}
\end{equation*}

It is easy to check that the value of $U$ and $V$ are always no less than 0 when $x' = 1$ and not all $x, l_2, \cdots, l_{\tau-1}$ are $0$ or $x' = x_1 = x_2= \cdots = x_{\kappa} = 1$ or $x' = x = l_2 = \cdots = l_{\tau-1} = 0$ and not all $x_1, x_2, \cdots, x_{\kappa}$ are 1. That means the inequality group (\ref{eq:improve1}) meets the conditions $1$, $2$ and $3$ of Table \ref{tb:x'}. When $x' = 1$, $x = l_2 =\cdots = l_{\tau-1} = 0$ and not all  $x_1, x_2, \cdots, x_{\kappa}$ are 1,  the maximum of $V$ is $-1$. When $x' = 0$ and not all $x, l_2, \cdots, l_{\tau-1}$ are $0$, the maximum of $U$ is $-1$. When $x' = 0$ and $x_1 = x_2 = \cdots = x_{\kappa} = 1$, the maximum of $U$ is $-1$. Therefore, all possible $x', x, l_2, \cdots, l_{\tau-1}, x_1, \cdots, x_\kappa$ satisfying the conditions 4, 5 and 6 in Table \ref{tb:x'} don't meet the inequality group (\ref{eq:improve1}). This completes the proof.  \hfill $\blacksquare$

\end{proof}

By Theorem \ref{th:3}, it is known that for the update of each state variable, two path variables and three inequalities are reduced from the inequality group.  As for SNOW 2.0 and Enocoro-128v2, the former is reduced totally $(4T + 32)\nu$ path variables and $(6T + 48)\nu$ inequalities, and the latter totally $(14T - 8)\nu$ path variables and $(21T - 12)\nu$ inequalities.

\subsection{The elimination of state variables and rules}

\subsubsection{Rules with two variables}$\\$

In this section we only consider the rules including two variables. Let $r$ be a rule with $x_1$ and $ x_2$. If $x_1$ and $x_2$ can be derived from each other, we always have $x_1 = x_2$. Therefore we can eliminate one of $x_1$ and $x_2$ and the rule $r$.

In fact, this improvement has been used in the inequality characterization of SNOW 2.0 in section 4.1, where totally 11 variables and 11 rules are reduced.

\subsubsection{Independent state variables and their rules}$\\$

In a deduction system, we call a variable $x$ to be independent if it is only used in a rule $r$. Here we assume that the rule $r$ contains exactly one independent variable and discuss it in two cases:
\begin{itemize}
	\item the independent variable $x$ belongs to a symmetrical rule $r$:
		\begin{equation}
			[x, x_{1}, x_{2}, \cdots, x_{\tau-1}].
		\end{equation}
		In this case, we can eliminate $x$ and r. That is because: 1) If $x$ is unknown and can be deduced by $r$, since $x$ only appears in $r$, it must be deduced by $r$ and the other variables in $r$ should be known initially or deduced by other rules. Obviously, when $x$ and $r$ are removed from the deduction system, it doesn't affect the other variables in $r$. 2) If $x$ is known in the initial state $X_0$, when $r$ is used, there exists a variable in $r$ expect $x$ deduced by $r$. We denote it by $x_1$. Take an initial state $X'_0$, where $x=0$, $x_1=1$ and the values of the others variables in $X'_0$ are the same as those in $X_0$. Obviously, the deduction course of $X_0$ and $X'_0$ is the same expect $r$. Note that $|X_0|$ = $|X'_0|$, where $|X|$ means the number of 1 in $X$,  if $X_0$ is a solution of the deduction system, $X'_0$ is also its solution, vice versa. Thus it does not change the minimum of the deduction system when $x$ and $r$ are removed.
		 
	\item the independent variable $x$ belongs to the following rule: 
		\begin{equation}
			x_{1}, x_{2}, \cdots, x_{\tau-1} \Rightarrow x_{\tau},
		\end{equation}
which means the right variables $x_{j}$ can be deduced by the left variables $x_{1}, x_{2}, \cdots, x_{\tau-1}$. If $x$ is one of the left variables, it must be known in the initial state $X_0$. If $x = x_{\tau}$, since $x$ only can be deduced by $r$, it does not affect the other deduction courses except $r$ when $x$ and $r$ are removed from the deduction system.
	
\end{itemize}

\newpage

 \appendix
 \renewcommand{\appendixname}{Appendix~\Alph{section}}

 \section{State variables and theirs paths of SNOW 2.0 with $T = 13$}
  \begin{table}[]\tiny\center
\begin{tabular}{|c|l|c|}
\hline
\multicolumn{1}{|c|}{State Var.} & \multicolumn{1}{c|}{Paths} & \multicolumn{1}{c|}{Num.}\\ \hline
$s_{0}^{(i+1)}$&$\{s_{0}^{(i)}\}$,$\{s_{16}^{(i)},s_{11}^{(i)},s_{2}^{(i)}\}$,$\{s_{15}^{(i)},R_{1}^{(i)},R_{0}^{(i)}\}$& 3\\ \hline
$s_{1}^{(i+1)}$&$\{s_{1}^{(i)}\}$,$\{s_{17}^{(i)},s_{12}^{(i)},s_{3}^{(i)}\}$,$\{s_{16}^{(i)},R_{2}^{(i)},R_{1}^{(i)}\}$& 3\\ \hline
$s_{2}^{(i+1)}$&$\{s_{2}^{(i)}\}$,$\{s_{16}^{(i)},s_{11}^{(i)},s_{0}^{(i)}\}$,$\{s_{18}^{(i)},s_{13}^{(i)},s_{4}^{(i)}\}$,$\{s_{17}^{(i)},R_{3}^{(i)},R_{2}^{(i)}\}$& 4\\ \hline
$s_{3}^{(i+1)}$&$\{s_{3}^{(i)}\}$,$\{s_{17}^{(i)},s_{12}^{(i)},s_{1}^{(i)}\}$,$\{s_{19}^{(i)},s_{14}^{(i)},s_{5}^{(i)}\}$,$\{s_{18}^{(i)},R_{4}^{(i)},R_{3}^{(i)}\}$& 4\\ \hline
$s_{4}^{(i+1)}$&$\{s_{4}^{(i)}\}$,$\{s_{18}^{(i)},s_{13}^{(i)},s_{2}^{(i)}\}$,$\{s_{20}^{(i)},s_{15}^{(i)},s_{6}^{(i)}\}$,$\{s_{19}^{(i)},R_{5}^{(i)},R_{4}^{(i)}\}$& 4\\ \hline
$s_{5}^{(i+1)}$&$\{s_{5}^{(i)}\}$,$\{s_{19}^{(i)},s_{14}^{(i)},s_{3}^{(i)}\}$,$\{s_{21}^{(i)},s_{16}^{(i)},s_{7}^{(i)}\}$,$\{s_{20}^{(i)},R_{6}^{(i)},R_{5}^{(i)}\}$,$\{R_{0}^{(i)},R_{2}^{(i)}\}$& 5\\ \hline
$s_{6}^{(i+1)}$&$\{s_{6}^{(i)}\}$,$\{s_{20}^{(i)},s_{15}^{(i)},s_{4}^{(i)}\}$,$\{s_{22}^{(i)},s_{17}^{(i)},s_{8}^{(i)}\}$,$\{s_{21}^{(i)},R_{7}^{(i)},R_{6}^{(i)}\}$,$\{R_{1}^{(i)},R_{3}^{(i)}\}$& 5\\ \hline
$s_{7}^{(i+1)}$&$\{s_{7}^{(i)}\}$,$\{s_{21}^{(i)},s_{16}^{(i)},s_{5}^{(i)}\}$,$\{s_{23}^{(i)},s_{18}^{(i)},s_{9}^{(i)}\}$,$\{s_{22}^{(i)},R_{8}^{(i)},R_{7}^{(i)}\}$,$\{R_{2}^{(i)},R_{4}^{(i)}\}$& 5\\ \hline
$s_{8}^{(i+1)}$&$\{s_{8}^{(i)}\}$,$\{s_{22}^{(i)},s_{17}^{(i)},s_{6}^{(i)}\}$,$\{s_{24}^{(i)},s_{19}^{(i)},s_{10}^{(i)}\}$,$\{s_{23}^{(i)},R_{9}^{(i)},R_{8}^{(i)}\}$,$\{R_{3}^{(i)},R_{5}^{(i)}\}$& 5\\ \hline
$s_{9}^{(i+1)}$&$\{s_{9}^{(i)}\}$,$\{s_{23}^{(i)},s_{18}^{(i)},s_{7}^{(i)}\}$,$\{s_{25}^{(i)},s_{20}^{(i)},s_{11}^{(i)}\}$,$\{s_{24}^{(i)},R_{10}^{(i)},R_{9}^{(i)}\}$,$\{R_{4}^{(i)},R_{6}^{(i)}\}$& 5\\ \hline
$s_{10}^{(i+1)}$&$\{s_{10}^{(i)}\}$,$\{s_{24}^{(i)},s_{19}^{(i)},s_{8}^{(i)}\}$,$\{s_{26}^{(i)},s_{21}^{(i)},s_{12}^{(i)}\}$,$\{s_{25}^{(i)},R_{11}^{(i)},R_{10}^{(i)}\}$,$\{R_{5}^{(i)},R_{7}^{(i)}\}$& 5\\ \hline
$s_{11}^{(i+1)}$&$\{s_{11}^{(i)}\}$,$\{s_{16}^{(i)},s_{2}^{(i)},s_{0}^{(i)}\}$,$\{s_{25}^{(i)},s_{20}^{(i)},s_{9}^{(i)}\}$,$\{s_{27}^{(i)},s_{22}^{(i)},s_{13}^{(i)}\}$,$\{s_{26}^{(i)},R_{12}^{(i)},R_{11}^{(i)}\}$,$\{R_{6}^{(i)},R_{8}^{(i)}\}$& 6\\ \hline
$s_{12}^{(i+1)}$&$\{s_{12}^{(i)}\}$,$\{s_{17}^{(i)},s_{3}^{(i)},s_{1}^{(i)}\}$,$\{s_{26}^{(i)},s_{21}^{(i)},s_{10}^{(i)}\}$,$\{s_{27}^{(i)},R_{13}^{(i)},R_{12}^{(i)}\}$,$\{R_{7}^{(i)},R_{9}^{(i)}\}$& 5\\ \hline
$s_{13}^{(i+1)}$&$\{s_{13}^{(i)}\}$,$\{s_{18}^{(i)},s_{4}^{(i)},s_{2}^{(i)}\}$,$\{s_{27}^{(i)},s_{22}^{(i)},s_{11}^{(i)}\}$,$\{R_{8}^{(i)},R_{10}^{(i)}\}$& 4\\ \hline
$s_{14}^{(i+1)}$&$\{s_{14}^{(i)}\}$,$\{s_{19}^{(i)},s_{5}^{(i)},s_{3}^{(i)}\}$,$\{R_{9}^{(i)},R_{11}^{(i)}\}$& 3\\ \hline
$s_{15}^{(i+1)}$&$\{s_{15}^{(i)}\}$,$\{s_{20}^{(i)},s_{6}^{(i)},s_{4}^{(i)}\}$,$\{R_{1}^{(i)},R_{0}^{(i)},s_{0}^{(i)}\}$,$\{R_{10}^{(i)},R_{12}^{(i)}\}$& 4\\ \hline
$s_{16}^{(i+1)}$&$\{s_{16}^{(i)}\}$,$\{s_{11}^{(i)},s_{2}^{(i)},s_{0}^{(i)}\}$,$\{s_{21}^{(i)},s_{7}^{(i)},s_{5}^{(i)}\}$,$\{R_{2}^{(i)},R_{1}^{(i)},s_{1}^{(i)}\}$,$\{R_{11}^{(i)},R_{13}^{(i)}\}$& 5\\ \hline
$s_{17}^{(i+1)}$&$\{s_{17}^{(i)}\}$,$\{s_{12}^{(i)},s_{3}^{(i)},s_{1}^{(i)}\}$,$\{s_{22}^{(i)},s_{8}^{(i)},s_{6}^{(i)}\}$,$\{R_{3}^{(i)},R_{2}^{(i)},s_{2}^{(i)}\}$& 4\\ \hline
$s_{18}^{(i+1)}$&$\{s_{18}^{(i)}\}$,$\{s_{13}^{(i)},s_{4}^{(i)},s_{2}^{(i)}\}$,$\{s_{23}^{(i)},s_{9}^{(i)},s_{7}^{(i)}\}$,$\{R_{4}^{(i)},R_{3}^{(i)},s_{3}^{(i)}\}$& 4\\ \hline
$s_{19}^{(i+1)}$&$\{s_{19}^{(i)}\}$,$\{s_{14}^{(i)},s_{5}^{(i)},s_{3}^{(i)}\}$,$\{s_{24}^{(i)},s_{10}^{(i)},s_{8}^{(i)}\}$,$\{R_{5}^{(i)},R_{4}^{(i)},s_{4}^{(i)}\}$& 4\\ \hline
$s_{20}^{(i+1)}$&$\{s_{20}^{(i)}\}$,$\{s_{15}^{(i)},s_{6}^{(i)},s_{4}^{(i)}\}$,$\{s_{25}^{(i)},s_{11}^{(i)},s_{9}^{(i)}\}$,$\{R_{6}^{(i)},R_{5}^{(i)},s_{5}^{(i)}\}$& 4\\ \hline
$s_{21}^{(i+1)}$&$\{s_{21}^{(i)}\}$,$\{s_{16}^{(i)},s_{7}^{(i)},s_{5}^{(i)}\}$,$\{s_{26}^{(i)},s_{12}^{(i)},s_{10}^{(i)}\}$,$\{R_{7}^{(i)},R_{6}^{(i)},s_{6}^{(i)}\}$& 4\\ \hline
$s_{22}^{(i+1)}$&$\{s_{22}^{(i)}\}$,$\{s_{17}^{(i)},s_{8}^{(i)},s_{6}^{(i)}\}$,$\{s_{27}^{(i)},s_{13}^{(i)},s_{11}^{(i)}\}$,$\{R_{8}^{(i)},R_{7}^{(i)},s_{7}^{(i)}\}$& 4\\ \hline
$s_{23}^{(i+1)}$&$\{s_{23}^{(i)}\}$,$\{s_{18}^{(i)},s_{9}^{(i)},s_{7}^{(i)}\}$,$\{R_{9}^{(i)},R_{8}^{(i)},s_{8}^{(i)}\}$& 3\\ \hline
$s_{24}^{(i+1)}$&$\{s_{24}^{(i)}\}$,$\{s_{19}^{(i)},s_{10}^{(i)},s_{8}^{(i)}\}$,$\{R_{10}^{(i)},R_{9}^{(i)},s_{9}^{(i)}\}$& 3\\ \hline
$s_{25}^{(i+1)}$&$\{s_{25}^{(i)}\}$,$\{s_{20}^{(i)},s_{11}^{(i)},s_{9}^{(i)}\}$,$\{R_{11}^{(i)},R_{10}^{(i)},s_{10}^{(i)}\}$& 3\\ \hline
$s_{26}^{(i+1)}$&$\{s_{26}^{(i)}\}$,$\{s_{21}^{(i)},s_{12}^{(i)},s_{10}^{(i)}\}$,$\{R_{12}^{(i)},R_{11}^{(i)},s_{11}^{(i)}\}$& 3\\ \hline
$s_{27}^{(i+1)}$&$\{s_{27}^{(i)}\}$,$\{s_{22}^{(i)},s_{13}^{(i)},s_{11}^{(i)}\}$,$\{R_{13}^{(i)},R_{12}^{(i)},s_{12}^{(i)}\}$& 3\\ \hline
$R_{0}^{(i+1)}$&$\{R_{0}^{(i)}\}$,$\{s_{15}^{(i)},R_{1}^{(i)},s_{0}^{(i)}\}$,$\{s_{5}^{(i)},R_{2}^{(i)}\}$& 3\\ \hline
$R_{1}^{(i+1)}$&$\{R_{1}^{(i)}\}$,$\{s_{15}^{(i)},R_{0}^{(i)},s_{0}^{(i)}\}$,$\{s_{16}^{(i)},R_{2}^{(i)},s_{1}^{(i)}\}$,$\{s_{6}^{(i)},R_{3}^{(i)}\}$& 4\\ \hline
$R_{2}^{(i+1)}$&$\{R_{2}^{(i)}\}$,$\{s_{16}^{(i)},R_{1}^{(i)},s_{1}^{(i)}\}$,$\{s_{17}^{(i)},R_{3}^{(i)},s_{2}^{(i)}\}$,$\{s_{7}^{(i)},R_{4}^{(i)}\}$,$\{s_{5}^{(i)},R_{0}^{(i)}\}$& 5\\ \hline
$R_{3}^{(i+1)}$&$\{R_{3}^{(i)}\}$,$\{s_{17}^{(i)},R_{2}^{(i)},s_{2}^{(i)}\}$,$\{s_{18}^{(i)},R_{4}^{(i)},s_{3}^{(i)}\}$,$\{s_{8}^{(i)},R_{5}^{(i)}\}$,$\{s_{6}^{(i)},R_{1}^{(i)}\}$& 5\\ \hline
$R_{4}^{(i+1)}$&$\{R_{4}^{(i)}\}$,$\{s_{18}^{(i)},R_{3}^{(i)},s_{3}^{(i)}\}$,$\{s_{19}^{(i)},R_{5}^{(i)},s_{4}^{(i)}\}$,$\{s_{9}^{(i)},R_{6}^{(i)}\}$,$\{s_{7}^{(i)},R_{2}^{(i)}\}$& 5\\ \hline
$R_{5}^{(i+1)}$&$\{R_{5}^{(i)}\}$,$\{s_{19}^{(i)},R_{4}^{(i)},s_{4}^{(i)}\}$,$\{s_{20}^{(i)},R_{6}^{(i)},s_{5}^{(i)}\}$,$\{s_{10}^{(i)},R_{7}^{(i)}\}$,$\{s_{8}^{(i)},R_{3}^{(i)}\}$& 5\\ \hline
$R_{6}^{(i+1)}$&$\{R_{6}^{(i)}\}$,$\{s_{20}^{(i)},R_{5}^{(i)},s_{5}^{(i)}\}$,$\{s_{21}^{(i)},R_{7}^{(i)},s_{6}^{(i)}\}$,$\{s_{11}^{(i)},R_{8}^{(i)}\}$,$\{s_{9}^{(i)},R_{4}^{(i)}\}$& 5\\ \hline
$R_{7}^{(i+1)}$&$\{R_{7}^{(i)}\}$,$\{s_{21}^{(i)},R_{6}^{(i)},s_{6}^{(i)}\}$,$\{s_{22}^{(i)},R_{8}^{(i)},s_{7}^{(i)}\}$,$\{s_{12}^{(i)},R_{9}^{(i)}\}$,$\{s_{10}^{(i)},R_{5}^{(i)}\}$& 5\\ \hline
$R_{8}^{(i+1)}$&$\{R_{8}^{(i)}\}$,$\{s_{22}^{(i)},R_{7}^{(i)},s_{7}^{(i)}\}$,$\{s_{23}^{(i)},R_{9}^{(i)},s_{8}^{(i)}\}$,$\{s_{13}^{(i)},R_{10}^{(i)}\}$,$\{s_{11}^{(i)},R_{6}^{(i)}\}$& 5\\ \hline
$R_{9}^{(i+1)}$&$\{R_{9}^{(i)}\}$,$\{s_{23}^{(i)},R_{8}^{(i)},s_{8}^{(i)}\}$,$\{s_{24}^{(i)},R_{10}^{(i)},s_{9}^{(i)}\}$,$\{s_{14}^{(i)},R_{11}^{(i)}\}$,$\{s_{12}^{(i)},R_{7}^{(i)}\}$& 5\\ \hline
$R_{10}^{(i+1)}$&$\{R_{10}^{(i)}\}$,$\{s_{24}^{(i)},R_{9}^{(i)},s_{9}^{(i)}\}$,$\{s_{25}^{(i)},R_{11}^{(i)},s_{10}^{(i)}\}$,$\{s_{15}^{(i)},R_{12}^{(i)}\}$,$\{s_{13}^{(i)},R_{8}^{(i)}\}$& 5\\ \hline
$R_{11}^{(i+1)}$&$\{R_{11}^{(i)}\}$,$\{s_{25}^{(i)},R_{10}^{(i)},s_{10}^{(i)}\}$,$\{s_{26}^{(i)},R_{12}^{(i)},s_{11}^{(i)}\}$,$\{s_{16}^{(i)},R_{13}^{(i)}\}$,$\{s_{14}^{(i)},R_{9}^{(i)}\}$& 5\\ \hline
$R_{12}^{(i+1)}$&$\{R_{12}^{(i)}\}$,$\{s_{26}^{(i)},R_{11}^{(i)},s_{11}^{(i)}\}$,$\{s_{27}^{(i)},R_{13}^{(i)},s_{12}^{(i)}\}$,$\{s_{15}^{(i)},R_{10}^{(i)}\}$& 4\\ \hline
$R_{13}^{(i+1)}$&$\{R_{13}^{(i)}\}$,$\{s_{27}^{(i)},R_{12}^{(i)},s_{12}^{(i)}\}$,$\{s_{16}^{(i)},R_{11}^{(i)}\}$& 3\\ \hline
\end{tabular}
\end{table}

\section{State variables and theirs paths of ENOCORO-128v2 with $T = 16$}
  \begin{table}[]\tiny\center
\begin{tabular}{|c|l|c|}
\hline
\multicolumn{1}{|c|}{State Var.} & \multicolumn{1}{c|}{Paths} & \multicolumn{1}{c|}{Num.}\\ \hline
$a_{0}^{(i+1)}$&$\{a_{0}^{(i)}\}$,$\{b_{3}^{(i)},e_{0}^{(i)}\}$,$\{f_{0}^{(i)},b_{0}^{(i)}\}$& 3\\ \hline
$a_{1}^{(i+1)}$&$\{a_{1}^{(i)}\}$,$\{b_{4}^{(i)},e_{1}^{(i)}\}$,$\{f_{1}^{(i)},b_{1}^{(i)}\}$,$\{g_{0}^{(i)},d_{0}^{(i)}\}$& 4\\ \hline
$a_{2}^{(i+1)}$&$\{a_{2}^{(i)}\}$,$\{b_{5}^{(i)},e_{2}^{(i)}\}$,$\{f_{2}^{(i)},b_{2}^{(i)}\}$,$\{g_{1}^{(i)},d_{1}^{(i)}\}$& 4\\ \hline
$a_{3}^{(i+1)}$&$\{a_{3}^{(i)}\}$,$\{b_{6}^{(i)},e_{3}^{(i)}\}$,$\{f_{3}^{(i)},b_{3}^{(i)}\}$,$\{g_{2}^{(i)},d_{2}^{(i)}\}$& 4\\ \hline
$a_{4}^{(i+1)}$&$\{a_{4}^{(i)}\}$,$\{b_{7}^{(i)},e_{4}^{(i)}\}$,$\{f_{4}^{(i)},b_{4}^{(i)}\}$,$\{g_{3}^{(i)},d_{3}^{(i)}\}$& 4\\ \hline
$a_{5}^{(i+1)}$&$\{a_{5}^{(i)}\}$,$\{b_{8}^{(i)},e_{5}^{(i)}\}$,$\{f_{5}^{(i)},b_{5}^{(i)}\}$,$\{g_{4}^{(i)},d_{4}^{(i)}\}$& 4\\ \hline
$a_{6}^{(i+1)}$&$\{a_{6}^{(i)}\}$,$\{b_{9}^{(i)},e_{6}^{(i)}\}$,$\{f_{6}^{(i)},b_{6}^{(i)}\}$,$\{g_{5}^{(i)},d_{5}^{(i)}\}$& 4\\ \hline
$a_{7}^{(i+1)}$&$\{a_{7}^{(i)}\}$,$\{b_{10}^{(i)},e_{7}^{(i)}\}$,$\{f_{7}^{(i)},b_{7}^{(i)}\}$,$\{g_{6}^{(i)},d_{6}^{(i)}\}$& 4\\ \hline
$a_{8}^{(i+1)}$&$\{a_{8}^{(i)}\}$,$\{b_{11}^{(i)},e_{8}^{(i)}\}$,$\{f_{8}^{(i)},b_{8}^{(i)}\}$,$\{g_{7}^{(i)},d_{7}^{(i)}\}$& 4\\ \hline
$a_{9}^{(i+1)}$&$\{a_{9}^{(i)}\}$,$\{b_{12}^{(i)},e_{9}^{(i)}\}$,$\{f_{9}^{(i)},b_{9}^{(i)}\}$,$\{g_{8}^{(i)},d_{8}^{(i)}\}$& 4\\ \hline
$a_{10}^{(i+1)}$&$\{a_{10}^{(i)}\}$,$\{b_{13}^{(i)},e_{10}^{(i)}\}$,$\{f_{10}^{(i)},b_{10}^{(i)}\}$,$\{g_{9}^{(i)},d_{9}^{(i)}\}$& 4\\ \hline
$a_{11}^{(i+1)}$&$\{a_{11}^{(i)}\}$,$\{b_{14}^{(i)},e_{11}^{(i)}\}$,$\{f_{11}^{(i)},b_{11}^{(i)}\}$,$\{g_{10}^{(i)},d_{10}^{(i)}\}$& 4\\ \hline
$a_{12}^{(i+1)}$&$\{a_{12}^{(i)}\}$,$\{f_{12}^{(i)},b_{12}^{(i)}\}$,$\{g_{11}^{(i)},d_{11}^{(i)}\}$& 3\\ \hline
$a_{13}^{(i+1)}$&$\{a_{13}^{(i)}\}$,$\{f_{13}^{(i)},b_{13}^{(i)}\}$,$\{g_{12}^{(i)},d_{12}^{(i)}\}$& 3\\ \hline
$a_{14}^{(i+1)}$&$\{a_{14}^{(i)}\}$,$\{f_{14}^{(i)},b_{14}^{(i)}\}$,$\{g_{13}^{(i)},d_{13}^{(i)}\}$& 3\\ \hline
$a_{15}^{(i+1)}$&$\{a_{15}^{(i)}\}$,$\{g_{14}^{(i)},d_{14}^{(i)}\}$& 2\\ \hline
$b_{0}^{(i+1)}$&$\{b_{0}^{(i)}\}$,$\{c_{5}^{(i)},c_{1}^{(i)}\}$,$\{f_{0}^{(i)},a_{0}^{(i)}\}$& 3\\ \hline
$b_{1}^{(i+1)}$&$\{b_{1}^{(i)}\}$,$\{c_{6}^{(i)},c_{2}^{(i)}\}$,$\{f_{1}^{(i)},a_{1}^{(i)}\}$& 3\\ \hline
$b_{2}^{(i+1)}$&$\{b_{2}^{(i)}\}$,$\{c_{7}^{(i)},c_{3}^{(i)}\}$,$\{f_{2}^{(i)},a_{2}^{(i)}\}$& 3\\ \hline
$b_{3}^{(i+1)}$&$\{b_{3}^{(i)}\}$,$\{a_{0}^{(i)},e_{0}^{(i)}\}$,$\{c_{8}^{(i)},c_{4}^{(i)}\}$,$\{f_{3}^{(i)},a_{3}^{(i)}\}$& 4\\ \hline
$b_{4}^{(i+1)}$&$\{b_{4}^{(i)}\}$,$\{a_{1}^{(i)},e_{1}^{(i)}\}$,$\{c_{9}^{(i)},c_{5}^{(i)}\}$,$\{f_{4}^{(i)},a_{4}^{(i)}\}$& 4\\ \hline
$b_{5}^{(i+1)}$&$\{b_{5}^{(i)}\}$,$\{a_{2}^{(i)},e_{2}^{(i)}\}$,$\{c_{10}^{(i)},c_{6}^{(i)}\}$,$\{f_{5}^{(i)},a_{5}^{(i)}\}$& 4\\ \hline
$b_{6}^{(i+1)}$&$\{b_{6}^{(i)}\}$,$\{a_{3}^{(i)},e_{3}^{(i)}\}$,$\{c_{11}^{(i)},c_{7}^{(i)}\}$,$\{f_{6}^{(i)},a_{6}^{(i)}\}$& 4\\ \hline
$b_{7}^{(i+1)}$&$\{b_{7}^{(i)}\}$,$\{a_{4}^{(i)},e_{4}^{(i)}\}$,$\{c_{12}^{(i)},c_{8}^{(i)}\}$,$\{f_{7}^{(i)},a_{7}^{(i)}\}$& 4\\ \hline
$b_{8}^{(i+1)}$&$\{b_{8}^{(i)}\}$,$\{a_{5}^{(i)},e_{5}^{(i)}\}$,$\{c_{13}^{(i)},c_{9}^{(i)}\}$,$\{f_{8}^{(i)},a_{8}^{(i)}\}$& 4\\ \hline
$b_{9}^{(i+1)}$&$\{b_{9}^{(i)}\}$,$\{a_{6}^{(i)},e_{6}^{(i)}\}$,$\{c_{14}^{(i)},c_{10}^{(i)}\}$,$\{f_{9}^{(i)},a_{9}^{(i)}\}$& 4\\ \hline
$b_{10}^{(i+1)}$&$\{b_{10}^{(i)}\}$,$\{a_{7}^{(i)},e_{7}^{(i)}\}$,$\{f_{10}^{(i)},a_{10}^{(i)}\}$& 3\\ \hline
$b_{11}^{(i+1)}$&$\{b_{11}^{(i)}\}$,$\{a_{8}^{(i)},e_{8}^{(i)}\}$,$\{f_{11}^{(i)},a_{11}^{(i)}\}$& 3\\ \hline
$b_{12}^{(i+1)}$&$\{b_{12}^{(i)}\}$,$\{a_{9}^{(i)},e_{9}^{(i)}\}$,$\{f_{12}^{(i)},a_{12}^{(i)}\}$& 3\\ \hline
$b_{13}^{(i+1)}$&$\{b_{13}^{(i)}\}$,$\{a_{10}^{(i)},e_{10}^{(i)}\}$,$\{f_{13}^{(i)},a_{13}^{(i)}\}$& 3\\ \hline
$b_{14}^{(i+1)}$&$\{b_{14}^{(i)}\}$,$\{a_{11}^{(i)},e_{11}^{(i)}\}$,$\{f_{14}^{(i)},a_{14}^{(i)}\}$& 3\\ \hline
$c_{0}^{(i+1)}$&$\{c_{0}^{(i)}\}$,$\{d_{9}^{(i)},d_{1}^{(i)}\}$,$\{f_{0}^{(i)},g_{0}^{(i)}\}$,$\{g_{0}^{(i)},e_{2}^{(i)}\}$,$\{f_{0}^{(i)},e_{2}^{(i)}\}$& 5\\ \hline
$c_{1}^{(i+1)}$&$\{c_{1}^{(i)}\}$,$\{c_{5}^{(i)},b_{0}^{(i)}\}$,$\{d_{10}^{(i)},d_{2}^{(i)}\}$,$\{f_{1}^{(i)},g_{1}^{(i)}\}$,$\{g_{1}^{(i)},e_{3}^{(i)}\}$,$\{f_{1}^{(i)},e_{3}^{(i)}\}$& 6\\ \hline
$c_{2}^{(i+1)}$&$\{c_{2}^{(i)}\}$,$\{c_{6}^{(i)},b_{1}^{(i)}\}$,$\{d_{11}^{(i)},d_{3}^{(i)}\}$,$\{f_{2}^{(i)},g_{2}^{(i)}\}$,$\{g_{2}^{(i)},e_{4}^{(i)}\}$,$\{f_{2}^{(i)},e_{4}^{(i)}\}$& 6\\ \hline
$c_{3}^{(i+1)}$&$\{c_{3}^{(i)}\}$,$\{c_{7}^{(i)},b_{2}^{(i)}\}$,$\{d_{12}^{(i)},d_{4}^{(i)}\}$,$\{f_{3}^{(i)},g_{3}^{(i)}\}$,$\{g_{3}^{(i)},e_{5}^{(i)}\}$,$\{f_{3}^{(i)},e_{5}^{(i)}\}$& 6\\ \hline
$c_{4}^{(i+1)}$&$\{c_{4}^{(i)}\}$,$\{c_{8}^{(i)},b_{3}^{(i)}\}$,$\{d_{13}^{(i)},d_{5}^{(i)}\}$,$\{f_{4}^{(i)},g_{4}^{(i)}\}$,$\{g_{4}^{(i)},e_{6}^{(i)}\}$,$\{f_{4}^{(i)},e_{6}^{(i)}\}$& 6\\ \hline
$c_{5}^{(i+1)}$&$\{c_{5}^{(i)}\}$,$\{b_{0}^{(i)},c_{1}^{(i)}\}$,$\{c_{9}^{(i)},b_{4}^{(i)}\}$,$\{d_{14}^{(i)},d_{6}^{(i)}\}$,$\{f_{5}^{(i)},g_{5}^{(i)}\}$,$\{g_{5}^{(i)},e_{7}^{(i)}\}$,$\{f_{5}^{(i)},e_{7}^{(i)}\}$& 7\\ \hline
$c_{6}^{(i+1)}$&$\{c_{6}^{(i)}\}$,$\{b_{1}^{(i)},c_{2}^{(i)}\}$,$\{c_{10}^{(i)},b_{5}^{(i)}\}$,$\{f_{6}^{(i)},g_{6}^{(i)}\}$,$\{g_{6}^{(i)},e_{8}^{(i)}\}$,$\{f_{6}^{(i)},e_{8}^{(i)}\}$& 6\\ \hline
$c_{7}^{(i+1)}$&$\{c_{7}^{(i)}\}$,$\{b_{2}^{(i)},c_{3}^{(i)}\}$,$\{c_{11}^{(i)},b_{6}^{(i)}\}$,$\{f_{7}^{(i)},g_{7}^{(i)}\}$,$\{g_{7}^{(i)},e_{9}^{(i)}\}$,$\{f_{7}^{(i)},e_{9}^{(i)}\}$& 6\\ \hline
$c_{8}^{(i+1)}$&$\{c_{8}^{(i)}\}$,$\{b_{3}^{(i)},c_{4}^{(i)}\}$,$\{c_{12}^{(i)},b_{7}^{(i)}\}$,$\{f_{8}^{(i)},g_{8}^{(i)}\}$,$\{g_{8}^{(i)},e_{10}^{(i)}\}$,$\{f_{8}^{(i)},e_{10}^{(i)}\}$& 6\\ \hline
$c_{9}^{(i+1)}$&$\{c_{9}^{(i)}\}$,$\{b_{4}^{(i)},c_{5}^{(i)}\}$,$\{c_{13}^{(i)},b_{8}^{(i)}\}$,$\{f_{9}^{(i)},g_{9}^{(i)}\}$,$\{g_{9}^{(i)},e_{11}^{(i)}\}$,$\{f_{9}^{(i)},e_{11}^{(i)}\}$& 6\\ \hline
$c_{10}^{(i+1)}$&$\{c_{10}^{(i)}\}$,$\{b_{5}^{(i)},c_{6}^{(i)}\}$,$\{c_{14}^{(i)},b_{9}^{(i)}\}$,$\{f_{10}^{(i)},g_{10}^{(i)}\}$,$\{g_{10}^{(i)},e_{12}^{(i)}\}$,$\{f_{10}^{(i)},e_{12}^{(i)}\}$& 6\\ \hline
\end{tabular}
\end{table}

\begin{table}[]\tiny\center
\begin{tabular}{|c|l|c|}
\hline
\multicolumn{1}{|c|}{State Var.} & \multicolumn{1}{c|}{Paths} & \multicolumn{1}{c|}{Num.}\\ \hline
$c_{11}^{(i+1)}$&$\{c_{11}^{(i)}\}$,$\{b_{6}^{(i)},c_{7}^{(i)}\}$,$\{f_{11}^{(i)},g_{11}^{(i)}\}$,$\{g_{11}^{(i)},e_{13}^{(i)}\}$,$\{f_{11}^{(i)},e_{13}^{(i)}\}$& 5\\ \hline
$c_{12}^{(i+1)}$&$\{c_{12}^{(i)}\}$,$\{b_{7}^{(i)},c_{8}^{(i)}\}$,$\{f_{12}^{(i)},g_{12}^{(i)}\}$,$\{g_{12}^{(i)},e_{14}^{(i)}\}$,$\{f_{12}^{(i)},e_{14}^{(i)}\}$& 5\\ \hline
$c_{13}^{(i+1)}$&$\{c_{13}^{(i)}\}$,$\{b_{8}^{(i)},c_{9}^{(i)}\}$,$\{f_{13}^{(i)},g_{13}^{(i)}\}$,$\{g_{13}^{(i)},e_{15}^{(i)}\}$,$\{f_{13}^{(i)},e_{15}^{(i)}\}$& 5\\ \hline
$c_{14}^{(i+1)}$&$\{c_{14}^{(i)}\}$,$\{b_{9}^{(i)},c_{10}^{(i)}\}$,$\{f_{14}^{(i)},g_{14}^{(i)}\}$,$\{g_{14}^{(i)},e_{16}^{(i)}\}$,$\{f_{14}^{(i)},e_{16}^{(i)}\}$& 5\\ \hline
$d_{0}^{(i+1)}$&$\{d_{0}^{(i)}\}$,$\{e_{15}^{(i)},e_{3}^{(i)}\}$,$\{g_{0}^{(i)},a_{1}^{(i)}\}$& 3\\ \hline
$d_{1}^{(i+1)}$&$\{d_{1}^{(i)}\}$,$\{d_{9}^{(i)},c_{0}^{(i)}\}$,$\{e_{16}^{(i)},e_{4}^{(i)}\}$,$\{g_{1}^{(i)},a_{2}^{(i)}\}$& 4\\ \hline
$d_{2}^{(i+1)}$&$\{d_{2}^{(i)}\}$,$\{d_{10}^{(i)},c_{1}^{(i)}\}$,$\{g_{2}^{(i)},a_{3}^{(i)}\}$& 3\\ \hline
$d_{3}^{(i+1)}$&$\{d_{3}^{(i)}\}$,$\{d_{11}^{(i)},c_{2}^{(i)}\}$,$\{g_{3}^{(i)},a_{4}^{(i)}\}$& 3\\ \hline
$d_{4}^{(i+1)}$&$\{d_{4}^{(i)}\}$,$\{d_{12}^{(i)},c_{3}^{(i)}\}$,$\{g_{4}^{(i)},a_{5}^{(i)}\}$& 3\\ \hline
$d_{5}^{(i+1)}$&$\{d_{5}^{(i)}\}$,$\{d_{13}^{(i)},c_{4}^{(i)}\}$,$\{g_{5}^{(i)},a_{6}^{(i)}\}$& 3\\ \hline
$d_{6}^{(i+1)}$&$\{d_{6}^{(i)}\}$,$\{d_{14}^{(i)},c_{5}^{(i)}\}$,$\{g_{6}^{(i)},a_{7}^{(i)}\}$& 3\\ \hline
$d_{7}^{(i+1)}$&$\{d_{7}^{(i)}\}$,$\{g_{7}^{(i)},a_{8}^{(i)}\}$& 2\\ \hline
$d_{8}^{(i+1)}$&$\{d_{8}^{(i)}\}$,$\{g_{8}^{(i)},a_{9}^{(i)}\}$& 2\\ \hline
$d_{9}^{(i+1)}$&$\{d_{9}^{(i)}\}$,$\{c_{0}^{(i)},d_{1}^{(i)}\}$,$\{g_{9}^{(i)},a_{10}^{(i)}\}$& 3\\ \hline
$d_{10}^{(i+1)}$&$\{d_{10}^{(i)}\}$,$\{c_{1}^{(i)},d_{2}^{(i)}\}$,$\{g_{10}^{(i)},a_{11}^{(i)}\}$& 3\\ \hline
$d_{11}^{(i+1)}$&$\{d_{11}^{(i)}\}$,$\{c_{2}^{(i)},d_{3}^{(i)}\}$,$\{g_{11}^{(i)},a_{12}^{(i)}\}$& 3\\ \hline
$d_{12}^{(i+1)}$&$\{d_{12}^{(i)}\}$,$\{c_{3}^{(i)},d_{4}^{(i)}\}$,$\{g_{12}^{(i)},a_{13}^{(i)}\}$& 3\\ \hline
$d_{13}^{(i+1)}$&$\{d_{13}^{(i)}\}$,$\{c_{4}^{(i)},d_{5}^{(i)}\}$,$\{g_{13}^{(i)},a_{14}^{(i)}\}$& 3\\ \hline
$d_{14}^{(i+1)}$&$\{d_{14}^{(i)}\}$,$\{c_{5}^{(i)},d_{6}^{(i)}\}$,$\{g_{14}^{(i)},a_{15}^{(i)}\}$& 3\\ \hline
$e_{0}^{(i+1)}$&$\{e_{0}^{(i)}\}$,$\{b_{3}^{(i)},a_{0}^{(i)}\}$& 2\\ \hline
$e_{1}^{(i+1)}$&$\{e_{1}^{(i)}\}$,$\{b_{4}^{(i)},a_{1}^{(i)}\}$& 2\\ \hline
$e_{2}^{(i+1)}$&$\{e_{2}^{(i)}\}$,$\{b_{5}^{(i)},a_{2}^{(i)}\}$,$\{f_{0}^{(i)},g_{0}^{(i)}\}$,$\{g_{0}^{(i)},c_{0}^{(i)}\}$,$\{f_{0}^{(i)},c_{0}^{(i)}\}$& 5\\ \hline
$e_{3}^{(i+1)}$&$\{e_{3}^{(i)}\}$,$\{b_{6}^{(i)},a_{3}^{(i)}\}$,$\{e_{15}^{(i)},d_{0}^{(i)}\}$,$\{f_{1}^{(i)},g_{1}^{(i)}\}$,$\{g_{1}^{(i)},c_{1}^{(i)}\}$,$\{f_{1}^{(i)},c_{1}^{(i)}\}$& 6\\ \hline
$e_{4}^{(i+1)}$&$\{e_{4}^{(i)}\}$,$\{b_{7}^{(i)},a_{4}^{(i)}\}$,$\{e_{16}^{(i)},d_{1}^{(i)}\}$,$\{f_{2}^{(i)},g_{2}^{(i)}\}$,$\{g_{2}^{(i)},c_{2}^{(i)}\}$,$\{f_{2}^{(i)},c_{2}^{(i)}\}$& 6\\ \hline
$e_{5}^{(i+1)}$&$\{e_{5}^{(i)}\}$,$\{b_{8}^{(i)},a_{5}^{(i)}\}$,$\{f_{3}^{(i)},g_{3}^{(i)}\}$,$\{g_{3}^{(i)},c_{3}^{(i)}\}$,$\{f_{3}^{(i)},c_{3}^{(i)}\}$& 5\\ \hline
$e_{6}^{(i+1)}$&$\{e_{6}^{(i)}\}$,$\{b_{9}^{(i)},a_{6}^{(i)}\}$,$\{f_{4}^{(i)},g_{4}^{(i)}\}$,$\{g_{4}^{(i)},c_{4}^{(i)}\}$,$\{f_{4}^{(i)},c_{4}^{(i)}\}$& 5\\ \hline
$e_{7}^{(i+1)}$&$\{e_{7}^{(i)}\}$,$\{b_{10}^{(i)},a_{7}^{(i)}\}$,$\{f_{5}^{(i)},g_{5}^{(i)}\}$,$\{g_{5}^{(i)},c_{5}^{(i)}\}$,$\{f_{5}^{(i)},c_{5}^{(i)}\}$& 5\\ \hline
$e_{8}^{(i+1)}$&$\{e_{8}^{(i)}\}$,$\{b_{11}^{(i)},a_{8}^{(i)}\}$,$\{f_{6}^{(i)},g_{6}^{(i)}\}$,$\{g_{6}^{(i)},c_{6}^{(i)}\}$,$\{f_{6}^{(i)},c_{6}^{(i)}\}$& 5\\ \hline
$e_{9}^{(i+1)}$&$\{e_{9}^{(i)}\}$,$\{b_{12}^{(i)},a_{9}^{(i)}\}$,$\{f_{7}^{(i)},g_{7}^{(i)}\}$,$\{g_{7}^{(i)},c_{7}^{(i)}\}$,$\{f_{7}^{(i)},c_{7}^{(i)}\}$& 5\\ \hline
$e_{10}^{(i+1)}$&$\{e_{10}^{(i)}\}$,$\{b_{13}^{(i)},a_{10}^{(i)}\}$,$\{f_{8}^{(i)},g_{8}^{(i)}\}$,$\{g_{8}^{(i)},c_{8}^{(i)}\}$,$\{f_{8}^{(i)},c_{8}^{(i)}\}$& 5\\ \hline
$e_{11}^{(i+1)}$&$\{e_{11}^{(i)}\}$,$\{b_{14}^{(i)},a_{11}^{(i)}\}$,$\{f_{9}^{(i)},g_{9}^{(i)}\}$,$\{g_{9}^{(i)},c_{9}^{(i)}\}$,$\{f_{9}^{(i)},c_{9}^{(i)}\}$& 5\\ \hline
$e_{12}^{(i+1)}$&$\{e_{12}^{(i)}\}$,$\{f_{10}^{(i)},g_{10}^{(i)}\}$,$\{g_{10}^{(i)},c_{10}^{(i)}\}$,$\{f_{10}^{(i)},c_{10}^{(i)}\}$& 4\\ \hline
$e_{13}^{(i+1)}$&$\{e_{13}^{(i)}\}$,$\{f_{11}^{(i)},g_{11}^{(i)}\}$,$\{g_{11}^{(i)},c_{11}^{(i)}\}$,$\{f_{11}^{(i)},c_{11}^{(i)}\}$& 4\\ \hline
$e_{14}^{(i+1)}$&$\{e_{14}^{(i)}\}$,$\{f_{12}^{(i)},g_{12}^{(i)}\}$,$\{g_{12}^{(i)},c_{12}^{(i)}\}$,$\{f_{12}^{(i)},c_{12}^{(i)}\}$& 4\\ \hline
$e_{15}^{(i+1)}$&$\{e_{15}^{(i)}\}$,$\{d_{0}^{(i)},e_{3}^{(i)}\}$,$\{f_{13}^{(i)},g_{13}^{(i)}\}$,$\{g_{13}^{(i)},c_{13}^{(i)}\}$,$\{f_{13}^{(i)},c_{13}^{(i)}\}$& 5\\ \hline
$e_{16}^{(i+1)}$&$\{e_{16}^{(i)}\}$,$\{d_{1}^{(i)},e_{4}^{(i)}\}$,$\{f_{14}^{(i)},g_{14}^{(i)}\}$,$\{g_{14}^{(i)},c_{14}^{(i)}\}$,$\{f_{14}^{(i)},c_{14}^{(i)}\}$& 5\\ \hline
$f_{0}^{(i+1)}$&$\{f_{0}^{(i)}\}$,$\{a_{0}^{(i)},b_{0}^{(i)}\}$,$\{g_{0}^{(i)},e_{2}^{(i)}\}$,$\{g_{0}^{(i)},c_{0}^{(i)}\}$,$\{e_{2}^{(i)},c_{0}^{(i)}\}$& 5\\ \hline
$f_{1}^{(i+1)}$&$\{f_{1}^{(i)}\}$,$\{a_{1}^{(i)},b_{1}^{(i)}\}$,$\{g_{1}^{(i)},e_{3}^{(i)}\}$,$\{g_{1}^{(i)},c_{1}^{(i)}\}$,$\{e_{3}^{(i)},c_{1}^{(i)}\}$& 5\\ \hline
$f_{2}^{(i+1)}$&$\{f_{2}^{(i)}\}$,$\{a_{2}^{(i)},b_{2}^{(i)}\}$,$\{g_{2}^{(i)},e_{4}^{(i)}\}$,$\{g_{2}^{(i)},c_{2}^{(i)}\}$,$\{e_{4}^{(i)},c_{2}^{(i)}\}$& 5\\ \hline
$f_{3}^{(i+1)}$&$\{f_{3}^{(i)}\}$,$\{a_{3}^{(i)},b_{3}^{(i)}\}$,$\{g_{3}^{(i)},e_{5}^{(i)}\}$,$\{g_{3}^{(i)},c_{3}^{(i)}\}$,$\{e_{5}^{(i)},c_{3}^{(i)}\}$& 5\\ \hline
$f_{4}^{(i+1)}$&$\{f_{4}^{(i)}\}$,$\{a_{4}^{(i)},b_{4}^{(i)}\}$,$\{g_{4}^{(i)},e_{6}^{(i)}\}$,$\{g_{4}^{(i)},c_{4}^{(i)}\}$,$\{e_{6}^{(i)},c_{4}^{(i)}\}$& 5\\ \hline
$f_{5}^{(i+1)}$&$\{f_{5}^{(i)}\}$,$\{a_{5}^{(i)},b_{5}^{(i)}\}$,$\{g_{5}^{(i)},e_{7}^{(i)}\}$,$\{g_{5}^{(i)},c_{5}^{(i)}\}$,$\{e_{7}^{(i)},c_{5}^{(i)}\}$& 5\\ \hline
\end{tabular}
\end{table}

 \begin{table}[]\tiny\center
\begin{tabular}{|c|l|c|}
\hline
\multicolumn{1}{|c|}{State Var.} & \multicolumn{1}{c|}{Paths} & \multicolumn{1}{c|}{Num.}\\ \hline 
$f_{6}^{(i+1)}$&$\{f_{6}^{(i)}\}$,$\{a_{6}^{(i)},b_{6}^{(i)}\}$,$\{g_{6}^{(i)},e_{8}^{(i)}\}$,$\{g_{6}^{(i)},c_{6}^{(i)}\}$,$\{e_{8}^{(i)},c_{6}^{(i)}\}$& 5\\ \hline
$f_{7}^{(i+1)}$&$\{f_{7}^{(i)}\}$,$\{a_{7}^{(i)},b_{7}^{(i)}\}$,$\{g_{7}^{(i)},e_{9}^{(i)}\}$,$\{g_{7}^{(i)},c_{7}^{(i)}\}$,$\{e_{9}^{(i)},c_{7}^{(i)}\}$& 5\\ \hline
$f_{8}^{(i+1)}$&$\{f_{8}^{(i)}\}$,$\{a_{8}^{(i)},b_{8}^{(i)}\}$,$\{g_{8}^{(i)},e_{10}^{(i)}\}$,$\{g_{8}^{(i)},c_{8}^{(i)}\}$,$\{e_{10}^{(i)},c_{8}^{(i)}\}$& 5\\ \hline
$f_{9}^{(i+1)}$&$\{f_{9}^{(i)}\}$,$\{a_{9}^{(i)},b_{9}^{(i)}\}$,$\{g_{9}^{(i)},e_{11}^{(i)}\}$,$\{g_{9}^{(i)},c_{9}^{(i)}\}$,$\{e_{11}^{(i)},c_{9}^{(i)}\}$& 5\\ \hline
$f_{10}^{(i+1)}$&$\{f_{10}^{(i)}\}$,$\{a_{10}^{(i)},b_{10}^{(i)}\}$,$\{g_{10}^{(i)},e_{12}^{(i)}\}$,$\{g_{10}^{(i)},c_{10}^{(i)}\}$,$\{e_{12}^{(i)},c_{10}^{(i)}\}$& 5\\ \hline
$f_{11}^{(i+1)}$&$\{f_{11}^{(i)}\}$,$\{a_{11}^{(i)},b_{11}^{(i)}\}$,$\{g_{11}^{(i)},e_{13}^{(i)}\}$,$\{g_{11}^{(i)},c_{11}^{(i)}\}$,$\{e_{13}^{(i)},c_{11}^{(i)}\}$& 5\\ \hline
$f_{12}^{(i+1)}$&$\{f_{12}^{(i)}\}$,$\{a_{12}^{(i)},b_{12}^{(i)}\}$,$\{g_{12}^{(i)},e_{14}^{(i)}\}$,$\{g_{12}^{(i)},c_{12}^{(i)}\}$,$\{e_{14}^{(i)},c_{12}^{(i)}\}$& 5\\ \hline
$f_{13}^{(i+1)}$&$\{f_{13}^{(i)}\}$,$\{a_{13}^{(i)},b_{13}^{(i)}\}$,$\{g_{13}^{(i)},e_{15}^{(i)}\}$,$\{g_{13}^{(i)},c_{13}^{(i)}\}$,$\{e_{15}^{(i)},c_{13}^{(i)}\}$& 5\\ \hline
$f_{14}^{(i+1)}$&$\{f_{14}^{(i)}\}$,$\{a_{14}^{(i)},b_{14}^{(i)}\}$,$\{g_{14}^{(i)},e_{16}^{(i)}\}$,$\{g_{14}^{(i)},c_{14}^{(i)}\}$,$\{e_{16}^{(i)},c_{14}^{(i)}\}$& 5\\ \hline
$g_{0}^{(i+1)}$&$\{g_{0}^{(i)}\}$,$\{d_{0}^{(i)},a_{1}^{(i)}\}$,$\{f_{0}^{(i)},e_{2}^{(i)}\}$,$\{f_{0}^{(i)},c_{0}^{(i)}\}$,$\{e_{2}^{(i)},c_{0}^{(i)}\}$& 5\\ \hline
$g_{1}^{(i+1)}$&$\{g_{1}^{(i)}\}$,$\{d_{1}^{(i)},a_{2}^{(i)}\}$,$\{f_{1}^{(i)},e_{3}^{(i)}\}$,$\{f_{1}^{(i)},c_{1}^{(i)}\}$,$\{e_{3}^{(i)},c_{1}^{(i)}\}$& 5\\ \hline
$g_{2}^{(i+1)}$&$\{g_{2}^{(i)}\}$,$\{d_{2}^{(i)},a_{3}^{(i)}\}$,$\{f_{2}^{(i)},e_{4}^{(i)}\}$,$\{f_{2}^{(i)},c_{2}^{(i)}\}$,$\{e_{4}^{(i)},c_{2}^{(i)}\}$& 5\\ \hline
$g_{3}^{(i+1)}$&$\{g_{3}^{(i)}\}$,$\{d_{3}^{(i)},a_{4}^{(i)}\}$,$\{f_{3}^{(i)},e_{5}^{(i)}\}$,$\{f_{3}^{(i)},c_{3}^{(i)}\}$,$\{e_{5}^{(i)},c_{3}^{(i)}\}$& 5\\ \hline
$g_{4}^{(i+1)}$&$\{g_{4}^{(i)}\}$,$\{d_{4}^{(i)},a_{5}^{(i)}\}$,$\{f_{4}^{(i)},e_{6}^{(i)}\}$,$\{f_{4}^{(i)},c_{4}^{(i)}\}$,$\{e_{6}^{(i)},c_{4}^{(i)}\}$& 5\\ \hline
$g_{5}^{(i+1)}$&$\{g_{5}^{(i)}\}$,$\{d_{5}^{(i)},a_{6}^{(i)}\}$,$\{f_{5}^{(i)},e_{7}^{(i)}\}$,$\{f_{5}^{(i)},c_{5}^{(i)}\}$,$\{e_{7}^{(i)},c_{5}^{(i)}\}$& 5\\ \hline
$g_{6}^{(i+1)}$&$\{g_{6}^{(i)}\}$,$\{d_{6}^{(i)},a_{7}^{(i)}\}$,$\{f_{6}^{(i)},e_{8}^{(i)}\}$,$\{f_{6}^{(i)},c_{6}^{(i)}\}$,$\{e_{8}^{(i)},c_{6}^{(i)}\}$& 5\\ \hline
$g_{7}^{(i+1)}$&$\{g_{7}^{(i)}\}$,$\{d_{7}^{(i)},a_{8}^{(i)}\}$,$\{f_{7}^{(i)},e_{9}^{(i)}\}$,$\{f_{7}^{(i)},c_{7}^{(i)}\}$,$\{e_{9}^{(i)},c_{7}^{(i)}\}$& 5\\ \hline
$g_{8}^{(i+1)}$&$\{g_{8}^{(i)}\}$,$\{d_{8}^{(i)},a_{9}^{(i)}\}$,$\{f_{8}^{(i)},e_{10}^{(i)}\}$,$\{f_{8}^{(i)},c_{8}^{(i)}\}$,$\{e_{10}^{(i)},c_{8}^{(i)}\}$& 5\\ \hline
$g_{9}^{(i+1)}$&$\{g_{9}^{(i)}\}$,$\{d_{9}^{(i)},a_{10}^{(i)}\}$,$\{f_{9}^{(i)},e_{11}^{(i)}\}$,$\{f_{9}^{(i)},c_{9}^{(i)}\}$,$\{e_{11}^{(i)},c_{9}^{(i)}\}$& 5\\ \hline
$g_{10}^{(i+1)}$&$\{g_{10}^{(i)}\}$,$\{d_{10}^{(i)},a_{11}^{(i)}\}$,$\{f_{10}^{(i)},e_{12}^{(i)}\}$,$\{f_{10}^{(i)},c_{10}^{(i)}\}$,$\{e_{12}^{(i)},c_{10}^{(i)}\}$& 5\\ \hline
$g_{11}^{(i+1)}$&$\{g_{11}^{(i)}\}$,$\{d_{11}^{(i)},a_{12}^{(i)}\}$,$\{f_{11}^{(i)},e_{13}^{(i)}\}$,$\{f_{11}^{(i)},c_{11}^{(i)}\}$,$\{e_{13}^{(i)},c_{11}^{(i)}\}$& 5\\ \hline
$g_{12}^{(i+1)}$&$\{g_{12}^{(i)}\}$,$\{d_{12}^{(i)},a_{13}^{(i)}\}$,$\{f_{12}^{(i)},e_{14}^{(i)}\}$,$\{f_{12}^{(i)},c_{12}^{(i)}\}$,$\{e_{14}^{(i)},c_{12}^{(i)}\}$& 5\\ \hline
$g_{13}^{(i+1)}$&$\{g_{13}^{(i)}\}$,$\{d_{13}^{(i)},a_{14}^{(i)}\}$,$\{f_{13}^{(i)},e_{15}^{(i)}\}$,$\{f_{13}^{(i)},c_{13}^{(i)}\}$,$\{e_{15}^{(i)},c_{13}^{(i)}\}$& 5\\ \hline
$g_{14}^{(i+1)}$&$\{g_{14}^{(i)}\}$,$\{d_{14}^{(i)},a_{15}^{(i)}\}$,$\{f_{14}^{(i)},e_{16}^{(i)}\}$,$\{f_{14}^{(i)},c_{14}^{(i)}\}$,$\{e_{16}^{(i)},c_{14}^{(i)}\}$& 5\\ \hline
\end{tabular}
\end{table}
  

\begin{thebibliography}{2019}

\bibitem{DIVIDE} T. Siegenthaler, "Decrypting a class of stream ciphers using ciphertext only", IEEE Transaction on Computer, Vol. 34, pp.81–85, 1985. 

\bibitem{F} J. Golic, "Cryptanalysis of alleged A5 stream cipher", EUROCRYPT’97, LNCS 1233, pp.239-255, 1997.

\bibitem{RC4}K. R. Lars, M. Willi, P. Bart, R. Vincent and V. Sven. "Analysis methods for (alleged) RC4", ASIACRYPT’98, pp.327-341, 1998.

\bibitem{WORD} H. Philip and R. G. Gregory, "Exploiting multiples of the connection polynomial in word-oriented stream ciphers", ASIACRYPT 2000, pp.303-316, 2000.

\bibitem{SNOW1.0}E. Patrik and J. Thomas, "SNOW-a new stream cipher", First Open NESSIE Workshop, pp.167-168, 2000.

\bibitem{GD-SOBER}C. D. Christophe, "Guess and determine attack on SOBER", NESSIE Public Document, NES/DOC/KUL/WP5/010/a, http://www.cryptonessie.org, 2001.

\bibitem{SOBER}R. Greg and H. Philip, "The t-class of SOBER stream ciphers", Technical Report, QUALCOMM Australia, Suite 410, Birkenhead Point, Drummoyne NSW 2137, Australia, 1999.

\bibitem{HEURISTIC}H. Ahmad and T. Eghlidos, "Heuristic guess-and-determine attacks on stream ciphers", IET Information Security. Vol. 3, Issue 2, pp.66–73, 2009.

\bibitem{PRUNE}B. Charles, D. Patrick and F. A. Pierre, "Automatic search of attacks on round-reduced AES and applications", CRYPTO 2011, LNCS 6841, pp.169–187, 2011.

\bibitem{LFSR}P. Enes, "On guess and determine cryptanalysis of LFSR-based stream ciphers", IEEE Transactions on Information Theory, Vol. 55,  Issue 7, pp.3398-3406, 2009.

\bibitem{FILTER}Y. Z. Wei, P. Enes and Y. P. Hu, "Guess and determine attacks on filter generators-revisited", IEEE Transactions on Information Theory, Vol. 58, Issue 4 pp.2530-2539, 2012.

\bibitem{SOSEMANUK}X. T. Feng, J. Liu, Z. C. Zhou, C. K. Wu and D. G. Feng, "A byte-based guess and determine attack on SOSEMANUK", ASIACRYPT 2010, pp.146-157, 2010.

\bibitem{Rabbit}X. T. Feng, Z. Q. Shi, C. K. Wu and D. G. Feng, “On guess and determine analysis of Rabbit”, International Journal of Foundations of Computer Science, Vol. 22, pp.1283–1296, 2011.

\bibitem{A2U2}Z. Q. Shi, X. T. Feng, D. G. Feng and C. K. Wu, "A real-time Key Recovery Attack on the Lightweight Stream Cipher A2U2", Cryptology and Network Security, CANS 2012, pp.12-22, 2012

\bibitem{FASER}X. T. Feng and F. Zhang, "A realtime key recovery attack on the authenticated cipher FASER128", IACR Cryptology ePrint Archive, https://eprint.iacr.org/2014/258, 2014.

\bibitem{Sablier}X. T. Feng and F. Zhang, "Cryptanalysis on the authenticated cipher sablier", International Conference on Network and System Security, NSS 2014,  Vol. 8792, pp 198-208, 2014.

\bibitem{PANDA-S}X. T. Feng, F. Zhang and H. WANG, "A practical forgery and state recovery attack on the authenticated cipher PANDA-s", IACR Cryptology ePrint Archive, https://eprint.iacr.org/2014/325, 2014.

\bibitem{SATGD}Z. Oleg and K. Stepan, "An improved SAT-based guess-and-determine attack on the alternating step generator", International Conference on Information Security 2017, LNCS 10599, pp.21-38, 2017.

\bibitem{FIRST}M. Nicky, Q. J. Wang, D. W. Gu and P. Bart, "Differential and linear cryptanalysis using mixed-integer linear programming",  INSCRYPT 2011, LNCS 7537, pp.57–76, 2011. 

\bibitem{BCS}S. B. Wu and M. S. Wang, "Security Evaluation against Differential Cryptanalysis for Block Cipher Structures", IACR Cryptology ePrint Archive, https://eprint.iacr.org/2011/551, 2011.

\bibitem{S-bP}S. W. Sun, L. Hu, L. Song, Y. H. Xie and P. Wang, "Automatic security evaluation of block ciphers with S-bP structures against related-key differential attacks", International Conference on Information Security and Cryptology 2013, LNCS 8567, pp.39-51, 2013.

\bibitem{Characteristic}S. W. Sun, L. Hu, M. Q. Wang, P. Wang, K. X. Qiao, X.S. Ma, D. P. Shi, L. Song and K. Fu, "Towards finding the best characteristics of some bit-oriented block ciphers and automatic enumeration of (related-key) differential and linear characteristics with predefined properties", Technical report, Cryptology ePrint Archive, https://eprint.iacr.org/2014/747, 2014.

\bibitem{ASE}S. W. Sun, L. Hu, P. Wang, K. X. Qiao, X. S. Ma and L. Song, "Automatic security evaluation and (Related-key) differential characteristic search: application to SIMON, PRESENT, LBlock, DES(L) and other Bit-Oriented Block Ciphers", ASIACRYPT 2014, LNCS 8873, pp. 158–178, 2014.

\bibitem{ARX}K. Fu, M. Q. Wang, Y. H. Guo, S. W. Sun and L Hu, "MILP-based automatic search algorithms for differential and linear trails for speck", International Workshop on Fast Software Encryption 2016, LNCS 9783, pp.268-288, 2016.

\bibitem{INPOSSIBLE}T. T. Cui, S. T. Chen, K. T. Jia, K. Fu and M. Q. Wang, "New automatic search tool for impossible differentials and zero-correlation linear approximations", IACR Cryptology ePrint Archive 2016, https://eprint.iacr.org/2016/689, 2016.

\bibitem{INPOSSIBLE2}S. Yu and T. Yosuke, "New impossible differential search tool from design and cryptanalysis aspects revealing structural properties of several ciphers", EUROCRYPT 2017, LNCS 10212, pp.185–215, 2017.

\bibitem{MORUS}D. P. Shi, S. W. Sun, S. Yu, C. Y. Li and L. Hu, "Correlation of quadratic boolean functions: cryptanalysis of all versions of full MORUS", CRYPTO 2019, pp.180-209, 2019.

\bibitem{Integral}Z. J. Xiang, W. T. Zhang, Z. Z. Bao and D. D. Lin, "Applying MILP method to searching integral distinguishers based on division property for 6 lightweight block ciphers", ASIACRYPT 2016, LNCS 10031, pp.648–678, 2016.

\bibitem{MILP}C. Marco, C. Alessandro, P. Massimo and S. D. Yaroslav, "Solving the lexicographic mixed-integer linear programming problem using Branch-and-Bound and Grossone methodology", Communcations in Nonlinear Science and Numerical Simulation, Vol. 84, 105177, https://doi.org/10.1016/j.cnsns.2020.105177, 2020

\bibitem{Gurobi}Gurobi: http://www.gurobi.com/

\bibitem{Cplex}Cplex: https://www.ibm.com/analytics/cplex-optimizer

\bibitem{MiniSat}MiniSat: http://minisat.se


\bibitem{SNOW2.0}E. Patrik and J. Thomas, "A new version of the stream cipher SNOW", SAC 2002, LNCS 2595, pp.47–61, 2002.


\bibitem{SNOWGD}H. Philip and R. G. Gregory, "Guess-and-determine attacks on SNOW", SAC 2002, LNCS 2595, pp.37–46, 2002.

\bibitem{SNOWDis}C. Don, H. Shai and J. Charanjit, "Cryptanalysis of stream ciphers with linear masking", CRYPTO 2002, LNCS 2442, pp. 515–532, 2002.

\bibitem{Standard}ISO/IEC 18033–4: Information technology - Security techniques - Encryption algorithms - Part 4: Stream ciphers, 2011.

\bibitem{Rijndael}J. Daemen and V. Rijmen, ”The design of Rijndael”, Series on Information Security and Cryptography, Springer Verlag, ISBN 3-540-42580-2, 2002.


\bibitem{Enocoro1.0} D. Watanabe and T. Kaneko, “A construction of light weight Panama-like keystream generator", IEICE Technical report, ISEC2007-78, 2007.

\bibitem{Enocoro-128v1.1} K. Muto, D. Watanabe and T. Kaneko, “Strength evaluation of Enocoro-128 against LDA and its Improvement”, Symposium on Cryptography and Information Security, SCIS 2008, 4A1-1, 2008. 

\bibitem{Enocoro-128v2-1} D. Watanabe, K. Okamoto and T. Kaneko, "A Hardware-Oriented Light Weight Pseudo-Random Number Generator Enocoro-128v2", The Symposium on Cryptography and Information Security, 3D1–3, 2010. 

\bibitem{Enocoro-128v2-2}D. Watanabe, T. Owada, K. Okamoto, Y. Igarashi, and T. Kaneko, "Update on Enocoro Stream Cipher", ISITA, pp.778–783, 2010.

\bibitem{Enocoro-128v2_evaluation}M. Hell and T. Johansson, "Security evaluation of stream cipher Enocoro-128v2", CRYPTREC Technical Report, 2010.

\bibitem{GUESS18}K. Ideguchi and D. Watanabe. "Method of security evaluation of guess and determine attacks", The 2008 Symposium on Cryptography and Information Security, SCIS 2008, 3A1-4, 2008.

\end{thebibliography}
\end{document}